\documentclass{article}

\usepackage[preprint]{neurips_2026}

\usepackage[utf8]{inputenc} 
\usepackage[T1]{fontenc}    
\usepackage{hyperref}       
\usepackage{url}            
\usepackage{booktabs}       
\usepackage{amsfonts}       
\usepackage{nicefrac}       
\usepackage{microtype}      
\usepackage{xcolor}         

\usepackage{amsmath,amssymb,amsthm,mathtools}
\usepackage[english]{babel}
\usepackage{chngcntr}
\usepackage{centernot}

\usepackage{geometry}
\usepackage{enumitem}

\usepackage{booktabs}

\newtheorem{heprinciple}{Heuristic Principle}

\newtheorem{theorem}{Theorem}

\newtheorem{proposition}{Proposition}

\title{No Certificate, No Execution: Certified Traces as a Foundation for Trustworthy AI Agents}

\author{
Xiao-Yang Liu Yanglet$^1$, Xiaodong Wang$^1$, Agostino Capponi$^2$\\
$^1$SecureFinAI Lab, Electrical Engineering, Columbia University \\
$^2$ Industrial Engineering and Operations Research, Columbia University \\
Emails: \{XL2427, XW2008, AC3827\}@columbia.edu \\
}

\begin{document}

\maketitle

\begin{abstract}
We argue that trustworthy AI agents, especially in high-stakes and policy-governed domains, \textit{should make execution conditional on certified traces rather than rely only on stronger generative models, output-level guardrails, or post-hoc audits}. A generative agent may propose useful recommendations, tool calls, reports, or candidate actions, but \textit{generation is not permission:
an action may be computable yet impermissible, and individually permissible actions may
compose into an impermissible trace}. We formalize trustworthy agency through a
\textbf{Proposal--Certification--Execution (PCE)} architecture: a probabilistic generating machine $M_G$ proposes candidate  execution traces, a \textbf{Permissibility Machine} $M_\Pi$ certifies whether a proposed trace is permitted
under a given policy system $\Pi$, and execution proceeds only for certified traces. 
The resulting
executable trace language is
\(
L_{\mathrm{exec}} = L_G \cap L_{\mathrm{cert}}(M_\Pi).
\)
Here a trace is not merely a retrospective log: before execution, it is a structured
pre-execution record submitted for certification, specifying the intended steps,
supporting evidence, proposed tool calls, required approvals, replayable computations,
credentials, and execution conditions.
This perspective complements chain-of-thought monitorability: visible reasoning may help
detect misbehavior, but \textit{monitorability is not certifiability}, and a reasoning trajectory is
only one component of a broader execution trace. The resulting formal principle is simple: an agent-generated trace should be authorized for execution only when it
carries a checkable certificate witnessing permissibility under $\Pi$: \textbf{no certificate, no execution}. We develop certified traces and \textit{Permissibility
Machines} as foundations for trustworthy AI agents, connect trace certification to proof-carrying
execution, proof memory, privacy, and zero-knowledge certificates, and propose a research agenda for evaluating agents by what generated traces can be safely
certified for execution, not by output accuracy alone.
\end{abstract}

\section{Introduction}
\label{sec:introduction}

Large language model-powered systems are moving beyond chatbots toward autonomous agents that participate in real-world operational workflows. Nowadays, individuals and institutions are actively exploring agents that retrieve documents, call APIs, write code, update memory, compute quantities, interact with external environments, and submit actions for execution~\citep{yao2023react,schick2023toolformer,debenedetti2024agentdojo}. 
Such exploratory agents raise a severe AI safety problem: a bad recommendation is an informational failure, but a bad execution can expose data, move money, alter records, or create institutional liability. Therefore, the central question is no longer whether an output looks acceptable, but whether the proposed trace is authorized to execute.  Capability is not authority. 
In high-stakes and policy-governed domains, agent behaviors must satisfy preset rules for access, privacy, evidence, computation, auditability, authorization, and execution.

Recent industry deployments make this trend concrete. OpenAI and PwC are building finance agents
for CFO workflows spanning planning, forecasting, reporting, procurement, payments, treasury, tax,
and close, with explicit emphasis on governance, controls, and human oversight~\citep{openai2026pwcfinance}.
American Express's Agentic Commerce Experiences initiative introduces agentic payment
credentials for verified AI agents, while Anthropic's financial-services agents connect to market data,
research platforms, and internal financial systems under governed access controls~\citep{amex2026ace,
anthropic2026financeagents}. ICBC similarly reports large-model and agentic systems for settlement
finance, credit risk, and bank-wide AI infrastructure~\citep{icbc2025interimqa}.
These deployments show that \textit{the question is moving from whether agents can assist finance teams
to whether agent-initiated actions can be authorized, constrained, and audited}.

\paragraph{Central principle.}
Trustworthy AI agents should not rely only on stronger generative models, output-level guardrails, or post-hoc audits. In high-stakes, policy-governed domains, execution should be authorized only for traces that carry checkable certificates under an explicit policy system. In short: no certificate, no execution.

The progression studied in this paper is
\[
\text{computable action}
\rightarrow
\text{permissible action}
\rightarrow
\text{permissible trace}
\rightarrow
\text{certified trace}
\rightarrow
\text{executable trace}.
\]
In this progression, ``permissible trace'' refers to a trace that satisfies the policy system, while
``certified trace'' refers to a proposed trace whose permissibility has been witnessed by a checkable
certificate. Execution turns a certified proposed trace into a realized trace that can later be audited.
A model may generate a candidate action, yet the action may be impermissible. Even if each
individual action is permissible, the resulting sequence may form an impermissible trace.
E.g., each database query may be individually authorized, while repeated queries to a
protected resource may violate a query-budget or privacy policy. Each trade may pass a local
order check, while the sequence violates a cumulative exposure constraint. A final answer may be
factually correct, while the derivation trace relies on unauthorized retrieval, fabricated evidence,
inconsistent units, or unreplayable computation. These examples point to the same conclusion:
\textit{the primary unit of execution trust should not be the final output, but the trace by which the output,
recommendation, or action was produced}.

Recent work on chain-of-thought monitorability ~\citep{korbak2025chain,openai2025cotmonitorability} highlights a related opportunity: if models expose
human-readable reasoning, monitors may inspect those reasoning trajectories for signs of
misbehavior. This is an important
but partial step. A chain of thoughts is one possible observation channel inside a broader execution
trace. High-stakes agents also produce retrievals, tool calls, database queries, memory writes, data
flows, logs, computations, and execution calls. Therefore, the safety object for agents is \textit{not merely the
chain of thoughts; it is the chain of executions}.

We propose \emph{certified traces} as the central unit of execution trust for AI agents. A trace is not only
a retrospective log of completed actions. Before execution, it is a candidate execution record: a
structured sequence of intended states, actions, tool calls, retrievals, computations, evidence,
memory updates, approvals, credentials, and execution conditions. After execution, the realized trace
becomes the audit record. A policy system $\Pi$ defines which traces are permissible. A
\emph{Permissibility Machine} is a certifying role: it checks whether a proposed trace carries a certificate
witnessing permissibility under $\Pi$. The resulting architecture separates three roles: a generating
machine proposes candidate traces, a Permissibility Machine certifies or rejects them, and an
executor acts only on certified traces.

In enterprise settings, this reframes accountability around three operational questions: \textbf{who
authorized this step, what evidence supported the authorization, and can the action be reconstructed
after failure?} A certified trace is intended to make these questions technically checkable rather than
left implicit in a generic system log. \textbf{This is where technical certification meets institutional accountability: liability cannot be assigned meaningfully if authorization evidence, policy context, and execution lineage are not recoverable}.

Although finance provides our main running examples, the trace-certification problem is broader.
Scientific and healthcare agents make the same issue visible: tool-using systems in chemistry,
biomedical discovery, and healthcare may retrieve sensitive data, call expert tools, and trigger
high-stakes downstream actions~\citep{bran2024chemcrow,tang2025aiscientists}.

We do not claim to solve all problems in verification, governance, or AI safety. Rather,
we propose a minimal formal interface for thinking about agent execution. The key objects are:
a policy-induced permissible trace region, a certified trace region, an execution-authorization rule,
and a certification boundary separating generated traces that may execute from those that must be
rejected or escalated.

Finance is used as a running domain because its policies, audit requirements, privacy constraints, evidence obligations, and execution risks make the trace-certification problem especially visible. However, the framework is intended more broadly for high-stakes, policy-governed AI agents: systems that retrieve sensitive information, call external tools, modify records, make regulated recommendations, or trigger real-world actions.

The paper makes four major contributions: it formulates the separation between generation and permission
at the trace level; introduces a \textbf{Proposal--Certification--Execution (PCE)} architecture; defines
certification-boundary measures for generated impermissibility, unsafe certification, missed
certification, and executable yield; and outlines a research agenda around proof memory, policy
models, privacy-preserving certification, and certified-agent evaluation.

\paragraph{Why now?}
This question has become urgent because AI systems are shifting from answer generation to delegated execution. Modern agents increasingly combine language models with retrieval, memory, tool use, code execution, database access, and external APIs. In such systems, safety failures may occur before the final output is produced, and they may concern source authorization, data flow, computation replay, cumulative risk, or policy history. This shift makes output inspection insufficient as the primary governance object. The relevant object is the proposed execution trace and the evidence required to authorize it.

\noindent \textbf{Take-home message}. \textit{Monitorability asks whether available observations are sufficient to detect or suspect
safety-relevant behavior. Certifiability asks whether the trace contains enough checkable evidence
to establish that policy obligations are satisfied. Execution authorization asks whether certification
has succeeded, so that the executor is allowed to act. To summarize,
\[
\boxed{
\text{monitorable trace}
\;\not\equiv\;
\text{certifiable trace}
\;\not\equiv\;
\text{execution-authorized trace}}.
\]}

\paragraph{Organization.}
Section~\ref{sec:output-guardrails} explains why output guardrails and local action checks are
insufficient. Section~\ref{sec:axioms} provides minimal axioms for certifiable execution and the
basic execution safety theorem. Sections~\ref{sec:pce-architecture}--\ref{sec:executable-languages}
introduce the proposal--certification--execution architecture and the executable trace language.
Sections~\ref{sec:policy-proof-memory}--\ref{sec:certification-boundary} discuss policy models,
proof memory, and certification boundary measures. Section~\ref{sec:privacy-zk} studies privacy
and zero-knowledge trace certification. Section~\ref{sec:objections-agenda} concludes with
objections, limitations, and open problems.

\section{Why Output Guardrails Fall Short}
\label{sec:output-guardrails}

The purpose of this section is to motivate trace certification by showing that, for
agentic systems, the safety-relevant predicate is often a property of the execution
trace rather than a property of the final output alone.

Output-level guardrails inspect final answers, block prohibited text, or apply local checks
to individual tool calls. These mechanisms
are useful, but they are structurally limited for agentic systems. Recent agent
benchmarks and attacks show that tool-using agents fail not only through final text,
but also through retrieval, tool outputs, environment observations, memory, and
multi-step actions~\citep{ruan2024toolemu,debenedetti2024agentdojo,yang2024watchagents,zhang2024agentsafetybench}. Recent security surveys
make the same point at the systems level: agentic risks arise across memory, tool
execution, multi-agent coordination, ecosystem, and governance layers, and many
threats are temporally extended rather than single-step failures~\citep{chu2026lasm}.
Thus the safety-relevant violation may occur inside the trace even when the final
output appears \textit{safe}.

We give three illustrative examples:
\begin{itemize}[leftmargin=*, label=$\bullet$]
    \item \textbf{Privacy and access control are naturally trace-dependent}. A single query to a protected
database may be authorized, while repeated queries to password-related or identity-related records
within a time window may violate a query-budget policy. Similarly, a final answer may contain no
personally identifiable information, while the trace that produced it may have exposed sensitive data
through prompts, retrieval, logs, memory, or tool calls.
    \item \textbf{Payment authorization is not output classification.} A payment instruction may be well-formed and may even match a local rule, while the full trace lacks valid authorization evidence: the beneficiary account may have changed, the approval may be stale, the amount may exceed a contextual limit, or the credential may not be scoped to this transaction. The relevant question is not whether the payment text looks valid, but whether the execution trace carries checkable authorization evidence.
    \item \textbf{A generated financial claim may be correct as a sentence but impermissible as a derivation.} An agent may correctly state that a firm's gross margin improved while relying on fabricated citations, inconsistent reporting periods, wrong XBRL/iXBRL tags, incompatible units, or arithmetic that cannot be replayed. This issue is especially visible in XBRL-based financial report analysis, where LLM agents must retrieve filing facts, interpret XBRL concepts, align periods and units, and perform reliable calculations~\citep{han2024xbrlagent}. The issue is not only whether the final claim is true, but whether it was produced by a permissible and auditable derivation trace.
\end{itemize}
Table~\ref{tab:trace_challenges} summarizes the trace-level challenges illustrated by these examples. These examples illustrate why output-level safety and local action safety do not
imply trace-level permissibility:
\[
\text{acceptable final output} \centernot\Rightarrow \tau \in L_\Pi,
\qquad
\text{locally permitted actions} \centernot\Rightarrow \tau \in L_\Pi .
\]
The failure is not that guardrails and monitors are useless; the failure is that
they do not by themselves define an execution-authorization rule. Output
guardrails ask whether the final output appears acceptable. Local action checks
ask whether each step is individually allowed. Monitoring asks whether available
observations raise concern. Trace certification asks a stronger question: whether
the proposed execution trace contains enough checkable evidence to satisfy the
policy system before execution.

\begin{table}[t]
\centering
\small
\renewcommand{\arraystretch}{1.22}
\begin{tabular}{p{0.24\linewidth}|p{0.37\linewidth}|p{0.30\linewidth}}
\hline\hline
\textbf{Challenge}
&
\textbf{Why output checking is insufficient}
&
\textbf{Permissibility-machine response}
\\
\hline
Output-centric guardrails
&
The final answer may look safe while the trace contains unsafe retrieval,
tool use, memory access, or execution.
&
Certify the full trace \(\tau\), not merely the output.
\\
\hline
Correct answer, wrong derivation
&
A claim may be true but supported by fabricated citations, inconsistent periods,
or unreplayable computations.
&
Require evidence binding and computation replay.
\\
\hline
Source access control
&
A generated answer may rely on sources the agent was not authorized to access.
&
Certify retrieval permissions and source lineage.
\\
\hline
Auditability
&
A later reviewer must know what was done and under which policy version.
&
Maintain persistent proof memory:
trace lineage plus policy lineage.
\\
\hline
Privacy-preserving compliance
&
The final output may omit private information while the trace leaks it through
prompts, tools, logs, or memory.
&
Use trace-level privacy certification and, when needed, zero-knowledge proofs.
\\
\hline\hline
\end{tabular}
\caption{Trace-level challenges motivating permissibility certification.}
\label{tab:trace_challenges}
\end{table}

This gives a hierarchy of increasingly strong execution controls:
\[
\text{output checking}
\;\rightarrow\;
\text{local action checks}
\;\rightarrow\;
\text{monitoring}
\;\rightarrow\;
\text{trace certification}
\;\rightarrow\;
\text{certified execution}.
\]
Each step adds information, but only the last two provide an authorization
condition for execution.
Our point is not to discard monitoring, but to place it inside a stronger execution-governance
pipeline.
The formal obstruction is trace non-compositionality: policies may be satisfied
locally at each step while being violated by the sequence.

\begin{proposition}[Trace Non-Compositionality]
\label{prop:trace-noncompositionality}
There exist policy systems and traces whose individual actions are locally permissible, while the
full trace is impermissible under the policy system.
\end{proposition}

\paragraph{Interpretation.}
This existence claim does not say that every policy is trace-dependent. It says
that important policy classes are: query budgets, temporal obligations, purpose
limitations, cumulative exposure limits, approval-before-execution rules, and
computation-replay requirements cannot be certified from the final output alone.
They require trace-level evidence.

\paragraph{Proof sketch.}
Construct a query-budget policy. Each individual query to a protected resource may be locally
authorized and therefore pass the action-level check. The policy, however, may permit at most five
queries to that resource within a time window. A trace containing six individually authorized
queries violates the trace-level budget, although every action is permissible. A full proof is given in Appx.~\ref{app:proofs}.

\section{Minimal Axioms of Certifiable Execution}
\label{sec:axioms}

We state the minimal assumptions under which trace certification implies execution safety.
Let $\Sigma^*$ denote the universe of encoded finite trace records and $\tau\in\Sigma^*$ denote a proposed trace submitted for certification. This trace may include
intended states, actions, tool calls, retrieval events, evidence, approvals, credentials, computations,
memory updates, and execution conditions. After execution, the realized trace records the events
actually performed and should conform to the certified proposed trace.
\begin{itemize}[leftmargin=*, label=$\bullet$]
    \item \textbf{Axiom 1: Policy boundary.}
A policy system $\Pi$ induces a semantic permissible trace region
\[
L_\Pi \subseteq \Sigma^*.
\]
A trace $\tau$ is permissible under $\Pi$ when $\tau\in L_\Pi$. This axiom defines the boundary
between permissible and impermissible traces.
\item \textbf{Axiom 2: Sound certification.}
A Permissibility Machine may be incomplete, conservative, or unable to certify many permissible
traces, but certification must be sound:
\(
L_{\mathrm{cert}}(M_\Pi) \subseteq L_\Pi.
\)
Thus every certified trace is semantically permissible. This axiom distinguishes a certificate from
a confidence score, self-approval, or post-hoc explanation.
\item \textbf{Axiom 3: Execution authorization.}
Execution is allowed only after certification succeeds:
\(
\mathrm{Execute}(\tau)
\;\Longrightarrow\;
\tau\in L_{\mathrm{cert}}(M_\Pi).
\)
This axiom makes certification operational rather than advisory: the executor cannot act on an
uncertified trace.
\item \textbf{Axiom 4: Execution conformance.}
Let $\tau^{\mathrm{prop}}$ denote the certified proposed trace and
$\tau^{\mathrm{real}}$ denote the realized trace after execution.
Execution conformance requires that either
\[
\mathrm{Conform}(\tau^{\mathrm{real}},\tau^{\mathrm{prop}})=1,
\]
or the executor halts, rolls back, or escalates before irreversible action. Moreover, conformance
must preserve permissibility:
\[
\tau^{\mathrm{prop}}\in L_\Pi
\;\wedge\;
\mathrm{Conform}(\tau^{\mathrm{real}},\tau^{\mathrm{prop}})=1
~~~~\Longrightarrow~~~~
\tau^{\mathrm{real}}\in L_\Pi .
\]
This axiom separates authorization of a proposed trace from audit of the realized trace. Certification
authorizes what the agent is permitted to do, while conformance checks whether the agent actually did it.
\end{itemize}

\begin{theorem}[Execution Authorization Safety]
\label{thm:execution-safety}
Under sound certification and execution authorization, every certified pre-execution trace submitted to the executor is semantically permissible:
\[
\mathrm{Execute}(\tau)
\;\Longrightarrow\;
\tau\in L_\Pi.
\]
\end{theorem}

\begin{proof}
By execution authorization in Axiom 3,
\(
\mathrm{Execute}(\tau)
\;\Longrightarrow\;
\tau\in L_{\mathrm{cert}}(M_\Pi).
\)
By sound certification in Axiom 2, $L_{\mathrm{cert}}(M_\Pi)\subseteq L_\Pi$. Therefore $\tau\in L_\Pi$.
\end{proof}

\paragraph{Interpretation.}
The separation principle, \textit{generation is not permission}, explains why certification is needed.
The three axioms explain why certification makes execution safe. Policy defines the permissible
region; certification verifies that a trace lies inside it; execution authorization ensures that only
certified traces are acted upon. Proof memory, privacy-preserving certification, policy drift, and
escalation are not required to state this basic theorem, but they become essential when the policy
boundary is trace-dependent, dynamic, or only partially observable.

\begin{theorem}[Certified Execution with Conformance]
Under sound certification, execution authorization, and execution conformance,
every realized executed trace is permissible:
\[
\mathrm{Execute}(\tau^{\mathrm{prop}})=\tau^{\mathrm{real}}
\Longrightarrow
\tau^{\mathrm{real}}\in L_\Pi .
\]
\end{theorem}

\begin{proof}
Suppose $\mathrm{Execute}(\tau^{\mathrm{prop}})=\tau^{\mathrm{real}}$.
By execution authorization, execution of $\tau^{\mathrm{prop}}$ implies
\[
\tau^{\mathrm{prop}}\in L_{\mathrm{cert}}(M_\Pi).
\]
By sound certification,
\[
L_{\mathrm{cert}}(M_\Pi)\subseteq L_\Pi,
\]
so $\tau^{\mathrm{prop}}\in L_\Pi$. By execution conformance, if execution proceeds rather
than halting, rolling back, or escalating, then
\[
\mathrm{Conform}(\tau^{\mathrm{real}},\tau^{\mathrm{prop}})=1.
\]
Using the conformance-preservation condition,
\[
\tau^{\mathrm{prop}}\in L_\Pi
\;\wedge\;
\mathrm{Conform}(\tau^{\mathrm{real}},\tau^{\mathrm{prop}})=1
\Longrightarrow
\tau^{\mathrm{real}}\in L_\Pi .
\]
Therefore every realized executed trace is permissible.
\end{proof}

\paragraph{Interpretation.}

The basic authorization-safety theorem proves that a trace certified before execution is permissible.
The conformance theorem adds the execution layer: the realized trace must conform to the certified
pre-execution trace, or the system must halt, roll back, or escalate before completion. Thus
``no certificate, no execution'' should be paired with ``no conformance, no completion.''

\section{Proposal--Certification--Execution Architecture}
\label{sec:pce-architecture}

We now define the three roles underlying certifiable execution: generation, certification, and
execution authorization. The purpose of this architecture is to prevent capability from being
confused with permission. A system may be able to propose a trace, but execution should occur
only after the trace has been certified under the given policy system $\Pi$.
\begin{itemize}[leftmargin=*, label=$\bullet$]
\item \textbf{Generator.}
A generating machine proposes a candidate trace:
\(
M_G: x \mapsto \tau,
\)
where $x$ is an input, prompt, task, or context, and $\tau$ is a finite proposed execution trace.
At this stage, $\tau$ is not yet a retrospective log; it is the structured record of what the agent
intends to do and what evidence or authorization context supports that execution. The role of $M_G$ is generative: it determines which traces can be produced.

\item \textbf{Permissibility Machine.}
A Permissibility Machine evaluates a proposed trace under a policy system $\Pi$ and returns a
certificate, a rejection, or an escalation:
\[
M_\Pi:(\tau,\Pi)\mapsto\{\pi,\bot,\mathrm{Escalate}\},
\]
where $\pi$ denotes a checkable certificate, $\bot$ denotes rejection, and
$\mathrm{Escalate}$ denotes transfer to a stronger certifier, a human reviewer, or an external review
process. \textit{High-risk execution need not always require manual approval, but it must require
checkable authorization evidence. Human approval
may be one component of a certificate, but certification is not defined as human approval: the
essential requirement is checkable authorization evidence}. The role of $M_\Pi$ is certifying: it determines whether the proposed trace has enough
evidence to be authorized for execution.

\item \textbf{Executor.}
An executor acts only on certified traces:
\(
E:\tau\mapsto o,
\)
where $o$ is the executed outcome. The role of $E$ is operational: it realizes the consequences of a
trace only after certification succeeds.
\end{itemize}

\paragraph{Principle: separation of proposal and certification.}
In high-stakes settings, the generator should not certify its own execution trace using only its
own generated explanation. The generator may supply evidence, proposed justifications, tool logs,
or candidate proof objects, but certification must be performed by a logically or institutionally
separate mechanism: for example, a policy engine, verifier, human reviewer, external auditor,
cryptographic proof system, or supervisory control. This separation is essential because a fluent
self-explanation is not a certificate. The agent may propose; it may not grant itself authority.

\paragraph{Three operational questions.}
The architecture separates three questions:
\[
\begin{array}{lll}
\text{Generator } M_G: & x \mapsto \tau
& \text{Can this trace be produced?} \\[2mm]
\text{Certifier } M_\Pi: & \tau \mapsto \{\pi,\bot,\mathrm{Escalate}\}
& \text{Can this trace be certified under } \Pi? \\[2mm]
\text{Executor } E: & \tau \mapsto o
& \text{Has execution been authorized?}
\end{array}
\]


\begin{figure}[t]
    \centering
    \includegraphics[width=0.95\textwidth]{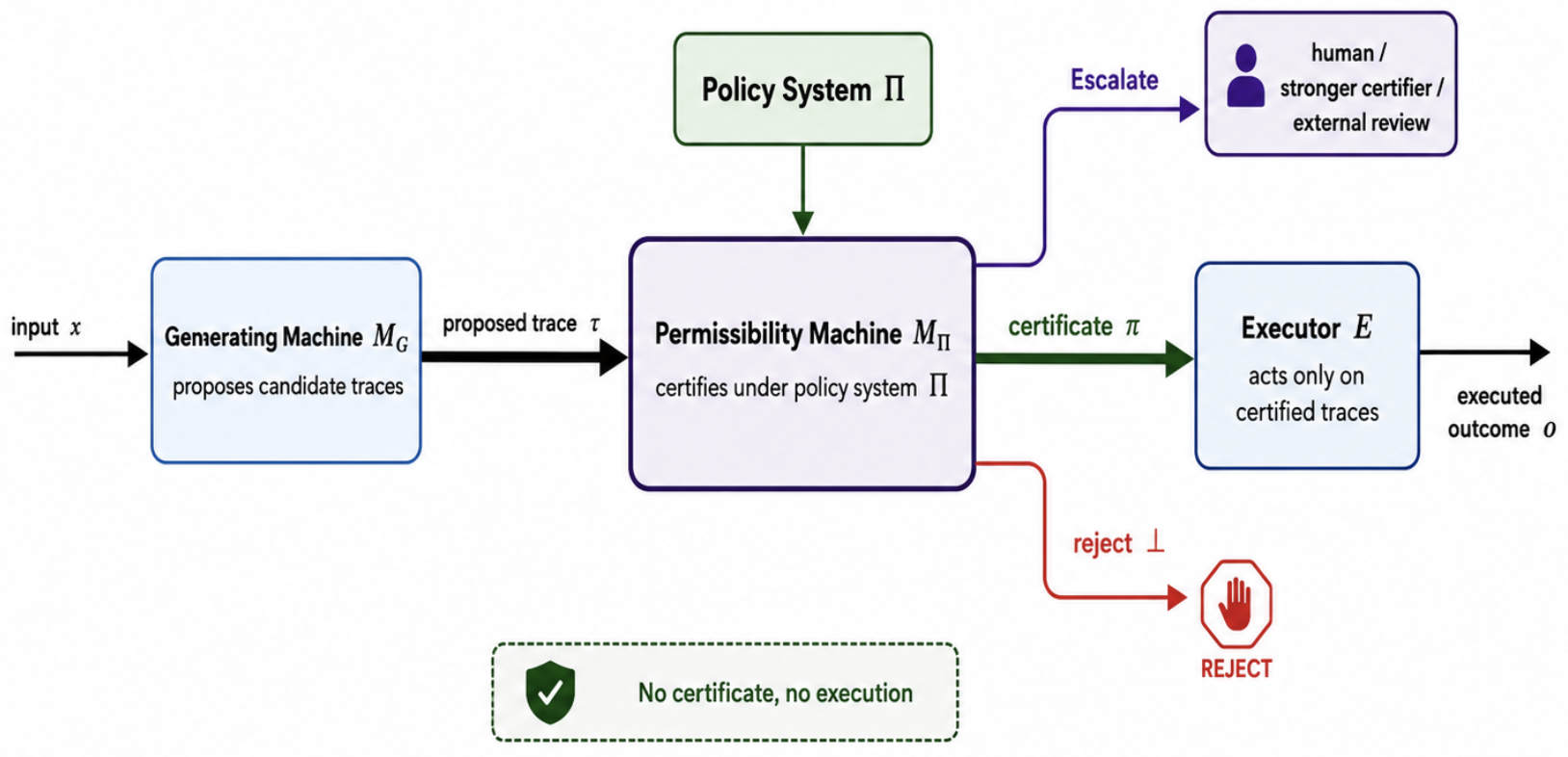} 
    \caption{The Proposal--Certification--Execution (PCE) architecture. A generating machine $M_G$ proposes
candidate traces, a Permissibility Machine $M_\Pi$ certifies, rejects, or escalates traces under a policy system $\Pi$, and execution is authorized only on the certificate path. The architecture enforces the invariant: no certificate, no execution.}
    \label{fig:pce-architecture}
    \vspace{-0.1in}
\end{figure}

\paragraph{A concrete instantiation.}
A Permissibility Machine need not be a single model or a single verifier. In a
financial-agent deployment, it may be implemented as a certifying service that
takes as input a proposed trace $\tau$, a policy system $\Pi$, proof memory $H$,
and external policy data $D_\Pi$, and returns
\[
M_\Pi(\tau,\Pi,H,D_\Pi)\in\{\pi,\bot,\mathrm{Escalate}\}.
\]
The proposed trace may include a task, retrieved sources, proposed tool calls,
computed quantities, approvals, credentials, risk checks, and intended execution
conditions. The policy system may include access-control rules, client mandates,
restricted lists, query budgets, exposure limits, approval-before-execution rules,
and computation-replay requirements. Proof memory stores the history needed for
certification, such as prior queries, current exposure, retrieved filings, policy
versions, approvals, and previous certificates.

For example, consider an agent that proposes to generate a client-facing financial
claim or initiate a payment instruction. The Permissibility Machine checks whether
the source was authorized, whether the data access matches the declared purpose,
whether required approvals are fresh and bound to the proposed trace, whether
risk or payment limits are satisfied, whether the computation is replayable, and
whether the credential is scoped to the intended action. If all required obligations
are witnessed, the machine emits a certificate
\[
\pi =
(\pi_{\mathrm{access}},
 \pi_{\mathrm{source}},
 \pi_{\mathrm{approval}},
 \pi_{\mathrm{risk}},
 \pi_{\mathrm{compute}},
 \pi_{\mathrm{lineage}}).
\]
If an obligation fails, it returns $\bot$; if evidence is incomplete or the action
exceeds the certifier's authority, it returns $\mathrm{Escalate}$. Thus the
Permissibility Machine is not an output classifier. It is an authorization
boundary: it decides whether the proposed trace carries enough checkable evidence
to be executed under $\Pi$.

This separation also clarifies the relation to monitoring and runtime enforcement. Runtime
verification monitors executions against specifications, while recent agent guardrails dynamically
check whether agent behavior satisfies safety requests~\citep{havelund2004jpaX,
xiang2024guardagent}. A Permissibility Machine plays a stronger role: it produces or verifies a
certificate required for execution authorization. Monitoring may provide evidence for certification,
but monitoring alone does not authorize execution.

\textbf{Architecture.}
A generating machine
proposes a candidate trace; a Permissibility Machine certifies, rejects, or escalates the trace under
the policy system; the executor acts only when certification succeeds. The operational chain is
\[
\boxed{
\text{Propose} \;\longrightarrow\; \text{Certify} \;\longrightarrow\; \text{Authorize Execution}}.
\]

Fig.~\ref{fig:pce-architecture} illustrates the
PCE architecture. The generating machine $M_G$ receives an input $x$ and proposes a candidate
trace $\tau$. The trace is then evaluated by the Permissibility Machine $M_\Pi$ under the policy
system $\Pi$. The certifier has three possible outcomes: it may return a certificate $\pi$, reject the
trace with $\bot$, or escalate the case to a human reviewer, stronger certifier, or external review
process. Execution is authorized only on the certificate path: the executor $E$ acts on $\tau$ only
when $M_\Pi(\tau,\Pi)=\pi$. Thus the figure encodes the central invariant of the architecture:
there is no direct execution path from generation to execution. Thus, no certificate, no execution!

\textbf{Design invariant}.
There is no direct execution path from $M_G$ to $E$. Execution is authorized only when
\(
M_\Pi(\tau,\Pi)=\pi.
\)
This invariant is an architectural expression of \emph{no certificate, no execution}!

\paragraph{Interpretation.}
The architecture does not assume that the generator is safe by construction, nor that the certifier
is complete. Its purpose is to separate proposal from authorization. The generator explores
candidate behavior; the Permissibility Machine governs admissibility under policy; 
the executor acts only on traces whose certification has succeeded.

\paragraph{Risk-tiered certification.}
The principle ``no certificate, no execution'' does not require the same certificate burden for every
agent action. Certification should be risk-tiered. Low-risk informational traces may require only
logging or lightweight provenance, while high-risk traces involving sensitive data, regulated advice,
external execution, asset transfer, or irreversible operations require stronger certificates. A useful
tiering is:
\[
\begin{array}{ll}
\mathrm{C0}: & \text{no external effect; logging only;}\\
\mathrm{C1}: & \text{low-risk informational response; lightweight provenance;}\\
\mathrm{C2}: & \text{client-facing or research claim; evidence and source certificate;}\\
\mathrm{C3}: & \text{sensitive data or tool access; access, purpose, and privacy certificate;}\\
\mathrm{C4}: & \text{regulated or irreversible action; authorization, policy, risk, and audit certificate;}\\
\mathrm{C5}: & \text{systemic or high-liability action; external, human, or regulator-grade review.}
\end{array}
\]
Thus certification is not a binary choice between no verification and full formal proof. The design
question is whether the certificate stack is proportionate to the execution risk of the trace.

\section{Executable Trace Languages}
\label{sec:executable-languages}

The PCE architecture turns execution authorization into a language-theoretic object. Let $\Sigma^*$ denote the universe of encoded finite trace records. The generator induces
a generated language $L_G \subseteq \Sigma^*$, consisting of the traces it can propose. For a policy system $\Pi$, the Permissibility Machine induces a certified language
$L_{\mathrm{cert}}(M_\Pi) \subseteq \Sigma^*$, consisting of proposed traces for which certification succeeds. Because execution is allowed only after certification, the executable
traces are precisely the generated traces that also lie in the certified language. This gives the central object of the paper: the executable trace language.

\textbf{Generated language.}
Let
\(
L_G := L(M_G) \subseteq \Sigma^*
\) 
denote the language of candidate execution traces proposed by the generating machine. Equivalently,
\[
L_G
=
\{\tau\in\Sigma^*:\exists x \text{ such that } M_G(x)=\tau\}.
\]

\textbf{Certified language.}
For a fixed policy system $\Pi$, let
\(
L_{\mathrm{cert}}(M_\Pi)
=
\{\tau\in\Sigma^*:M_\Pi(\tau,\Pi)=\pi\}
\)
denote the language of traces certified by the Permissibility Machine.

\textbf{Executable language.}
A proposed trace is authorized for execution only when certification succeeds. We therefore define
the executable trace language by
\[
L_{\mathrm{exec}}(\mathcal A)
=
L_G\cap L_{\mathrm{cert}}(M_\Pi),
\]
where
\(
\mathcal A=(M_G,M_\Pi,E)
\)
denotes the proposal--certification--execution architecture.

\begin{theorem}[Executable Language Composition]
\label{thm:exec-language-composition}
Fix a policy system $\Pi$ and architecture
\(
\mathcal A=(M_G,M_\Pi,E).
\)
Under execution authorization, a trace is executable if and only if it is both generated and certified.
Equivalently,
\(
L_{\mathrm{exec}}(\mathcal A)
=
L_G\cap L_{\mathrm{cert}}(M_\Pi).
\)
\end{theorem}

\textbf{Interpretation.}
The identity
\(
L_{\mathrm{exec}}(\mathcal A)=L_G\cap L_{\mathrm{cert}}(M_\Pi)
\)
captures the proposal--certification--execution logic in one line. Generation determines what can
be proposed. Certification determines which proposed traces are authorized for execution. Policy
determines the semantic target: the permissible trace language $L_\Pi$. Thus execution is not
identified with generation, and certification is not merely advisory; it is the authorization condition
between proposal and execution.

\textbf{Proof sketch.}
If a trace is executable, then it must have been produced by the generator, so $\tau\in L_G$.
Execution authorization also requires certification, so $\tau\in L_{\mathrm{cert}}(M_\Pi)$. Hence
\(
\tau\in L_G\cap L_{\mathrm{cert}}(M_\Pi).
\)
Conversely, if $\tau$ lies in this intersection, then it is a generated trace for which certification
has succeeded. Under the architecture, such traces are authorized for execution. A full proof is
given in Appx.~\ref{app:proofs}.

\textbf{Semantic versus operational boundaries.}
The permissible language $L_\Pi$ and the certified language $L_{\mathrm{cert}}(M_\Pi)$ are distinct
objects. The former is semantic: it is defined by the policy system. The latter is operational: it is
defined by what the certifier can recognize and certify. Soundness requires
\[
L_{\mathrm{cert}}(M_\Pi)\subseteq L_\Pi,
\]
but completeness need not hold. The mismatch between the semantic permissible boundary and the
operational certification boundary is the source of the evaluation quantities studied next.

\section{Policy Models and Proof Memory}
\label{sec:policy-proof-memory}

The policy system determines what must be certified; proof memory determines what can be
remembered, replayed, and certified. This distinction is essential. A policy system specifies the
obligations a trace must satisfy, while proof memory specifies the proof-relevant history available
to the certifier.

\textbf{Policy-system interface.}
At the semantic level, a policy system induces a permissible trace language
\[
\Pi \leadsto L_\Pi\subseteq\Sigma^* .
\]
A trace $\tau$ is permissible under $\Pi$ when $\tau\in L_\Pi$. For proof-carrying execution, inspired by proof-carrying code~\citep{necula1997proof}, the
policy system may also support a derivability judgment,
\(
\Pi\vdash \tau:\mathrm{Permitted},
\)
meaning that the permissibility of $\tau$ can be established within the policy calculus. A
certificate $\pi$ is a checkable witness for such a judgment.

\textbf{Canonical policy models.}
Different policy structures create different certification requirements. We use four canonical
models as running examples.
\begin{itemize}[leftmargin=*, label=$\bullet$]
    \item \textbf{Semantic trace-language policies} represent a policy only by its induced language
    $L_\Pi\subseteq\Sigma^*$. This model is sufficient for execution safety and
    certification-boundary analysis.

    \item \textbf{Counter or resource-budget policies} constrain cumulative quantities along the
    trace, such as exposure limits, query budgets, rate limits, or liquidity constraints. These
    policies are trace-dependent because the permissibility of the next action depends on
    accumulated history.

    \item \textbf{Temporal or history-dependent policies} impose ordering requirements, such as
    approval-before-execution, deletion after a retention window, mandatory escalation after
    unresolved evidence, or restrictions on reusing data for a new purpose.

    \item \textbf{Information-flow policies} govern access, source authorization, personally identifiable information
(PII), purpose limitation, logging, retention, and downstream release, following the broader
language-based information-flow tradition~\citep{sabelfeld2003informationflow}.
\end{itemize}

\textbf{Proof memory.}
Proof memory stores the history needed to certify these policies: prior actions, query counts,
retrieved sources, evidence, computed values, approvals, policy versions, access decisions, and
audit records. A certifier cannot enforce a five-query budget unless it remembers the relevant query
count; it cannot certify evidence binding unless it remembers sources, tags, periods, units, and
computations; and it cannot support later audit unless it records the policy version under which the
trace was certified.

\textbf{Worked certificate example.}
Consider an agent that produces a financial claim from a filing. A certificate for the trace may
contain proof objects
\[
\pi =
(\pi_{\mathrm{access}},
 \pi_{\mathrm{source}},
 \pi_{\mathrm{tag}},
 \pi_{\mathrm{period}},
 \pi_{\mathrm{unit}},
 \pi_{\mathrm{compute}},
 \pi_{\mathrm{claim}}).
\]
These components certify, respectively, that the source was authorized, the filing was retrieved,
the relevant data tag was used, the period was aligned, units were normalized, the computation was
replayable, and the final claim was supported by the derivation trace. The certificate is not a
natural-language explanation; it is a structured witness that the trace satisfies policy obligations. 

In financial reporting, XBRL/iXBRL provides a natural evidence layer for such certificates:
reported facts are associated with machine-readable concepts, reporting contexts, units, and filings~\citep{xbrl_standard,xbrl_ixbrl}. This connects directly to recent work on XBRL-based financial report analysis. XBRL-Agent shows that LLM agents augmented with retrievers and
calculators still face reliability challenges in XBRL term understanding and financial
calculation~\citep{han2024xbrlagent}. From the certified-trace perspective, the remedy is not only better retrieval or calculation, but a certificate binding the generated claim to authorized sources, XBRL concepts, reporting periods, units, replayable computations, and the final statement.

\paragraph{Certificate stack.}
In high-stakes deployments, a certificate should not be viewed as a single monolithic proof.
Rather, it is often a stack of proof objects, each witnessing a different policy obligation. We write
\[
\pi =
(\pi_{\mathrm{id}},
 \pi_{\mathrm{auth}},
 \pi_{\mathrm{source}},
 \pi_{\mathrm{policy}},
 \pi_{\mathrm{risk}},
 \pi_{\mathrm{compute}},
 \pi_{\mathrm{privacy}},
 \pi_{\mathrm{human}},
 \pi_{\mathrm{lineage}}).
\]
Here $\pi_{\mathrm{id}}$ identifies the agent, user, principal, or institutional role on whose behalf
the trace is proposed; $\pi_{\mathrm{auth}}$ certifies scoped authority or credentials;
$\pi_{\mathrm{source}}$ certifies source access and evidence provenance; $\pi_{\mathrm{policy}}$
records the applicable policy clauses; $\pi_{\mathrm{risk}}$ certifies risk limits such as exposure,
query budgets, payment thresholds, or trading constraints; $\pi_{\mathrm{compute}}$ certifies
replayable computation, including formulas, units, periods, and code; $\pi_{\mathrm{privacy}}$
certifies access, retention, release, and information-flow obligations; $\pi_{\mathrm{human}}$
records human approval or escalation when required; and $\pi_{\mathrm{lineage}}$ records the
policy version, certificate version, timestamp, and audit linkage.

Not every trace requires every component. The certificate stack is a design interface: the policy
system determines which components are mandatory for a given trace, while proof memory determines
which components can be produced, checked, or replayed.

\textbf{Expressiveness and capacity.}
Proof memory has both qualitative and quantitative aspects. \emph{Proof-memory expressiveness}
asks which trace properties or trace languages can be certified with a given memory model.
\emph{Proof-memory capacity} asks how many proof-relevant histories the certifier can distinguish
under resource constraints. The basic principle is
\[
\text{a certifier can only certify trace properties that its proof memory can distinguish.}
\]
Thus proof memory determines the resolution of the certification boundary.

\section{Certification Boundary and Evaluation Measures}
\label{sec:certification-boundary}

The framework above turns trustworthy agency into a measurable boundary problem. A generator
induces a distribution over proposed traces; the policy system defines which of those traces are
semantically permissible; the certifier determines which traces cross the certification boundary. This
section introduces the basic quantities needed to study that boundary.

\textbf{Trace universe and generator distribution.}
Let $\mu_G$ denote the generator-induced distribution over traces. We view $\mu_G$ as a
distribution over $\Sigma^*$ supported on generated traces, so generated behavior is measured
relative to the agent's own proposal distribution.

\textbf{Generated permissibility risk.}
The \emph{generated permissibility risk} is the probability that the generator proposes an
impermissible trace:
\[
\mathrm{Gap}_\Pi(M_G)
=
\Pr_{\tau\sim\mu_G}[\tau\notin L_\Pi].
\]
This quantity measures the semantic gap between generation and permissibility.

\textbf{Certified-but-impermissible risk.}
The \emph{certified-but-impermissible risk} is
\[
u_\Pi(M_G,M_\Pi)
=
\Pr_{\tau\sim\mu_G}
\bigl[
\tau\in L_{\mathrm{cert}}(M_\Pi)\setminus L_\Pi
\bigr].
\]
This is the unsafe-admission event: the certifier admits a trace that is in fact impermissible.

\textbf{Missed-certification risk.}
The \emph{missed-certification risk} is
\[
m_\Pi(M_G,M_\Pi)
=
\Pr_{\tau\sim\mu_G}
\bigl[
\tau\in L_\Pi\setminus L_{\mathrm{cert}}(M_\Pi)
\bigr].
\]
This is the incompleteness event: a semantically permissible trace fails to be certified.

\textbf{Executable yield.}
The \emph{executable yield} is the probability mass of certified traces:
\[
y_\Pi(M_G,M_\Pi)
=
\Pr_{\tau\sim\mu_G}
\bigl[
\tau\in L_{\mathrm{cert}}(M_\Pi)
\bigr].
\]
Because $\mu_G$ is generator-induced, this is precisely the fraction of generated trace mass that
crosses the certification boundary.

\begin{theorem}[Certification Boundary Decomposition]
\label{thm:boundary-decomposition}
For any policy system $\Pi$, generator $M_G$, and certifier $M_\Pi$, we have 
\(
y_\Pi
=
1-\mathrm{Gap}_\Pi-m_\Pi+u_\Pi.
\)
If the certifier is sound, then $u_\Pi=0$ and therefore
\(
y_\Pi
=
1-\mathrm{Gap}_\Pi-m_\Pi.
\)
\end{theorem}

\textbf{Interpretation.}
The theorem is a conservation law for generated trace mass. Relative to the generator distribution,
a trace is either permissible or impermissible. Among permissible traces, some are certified and
some are missed. Among impermissible traces, some are correctly rejected and some are wrongly
certified. The executable yield is therefore the permissible generated mass, minus missed
certification, plus unsafe certification.

\textbf{Boundary distance.}
A natural measure of the distance between the operational certification boundary and the semantic
permissible boundary is
\[
d_\mu(M_\Pi,\Pi)
=
\mu_G\!\left(L_{\mathrm{cert}}(M_\Pi)\triangle L_\Pi\right)
=
u_\Pi(M_G,M_\Pi)+m_\Pi(M_G,M_\Pi),
\]
where $\triangle$ denotes symmetric difference. This quantity makes explicit that certification
quality is a boundary-approximation problem: the certifier is good when its certified region closely
matches the true permissible region under the generator's trace distribution.

\textbf{Monitorability gap versus certification gap.}
Let $O(\tau)$ denote the observable part of a trace available to a monitor, such as outputs,
actions, tool logs, activations, or chain-of-thought. A monitorability gap arises when
$O(\tau)$ does not contain enough policy-relevant information to detect whether
$\tau\notin L_\Pi$. A certification gap arises when the available observations and proof memory
do not contain enough checkable evidence to produce a certificate, even if the trace is semantically
permissible. Thus monitorability and certifiability are related but distinct: monitorability concerns
detection from observations, while certifiability concerns checkable permission for execution.

\begin{figure}[t]
    \centering
    \includegraphics[width=0.70\textwidth]{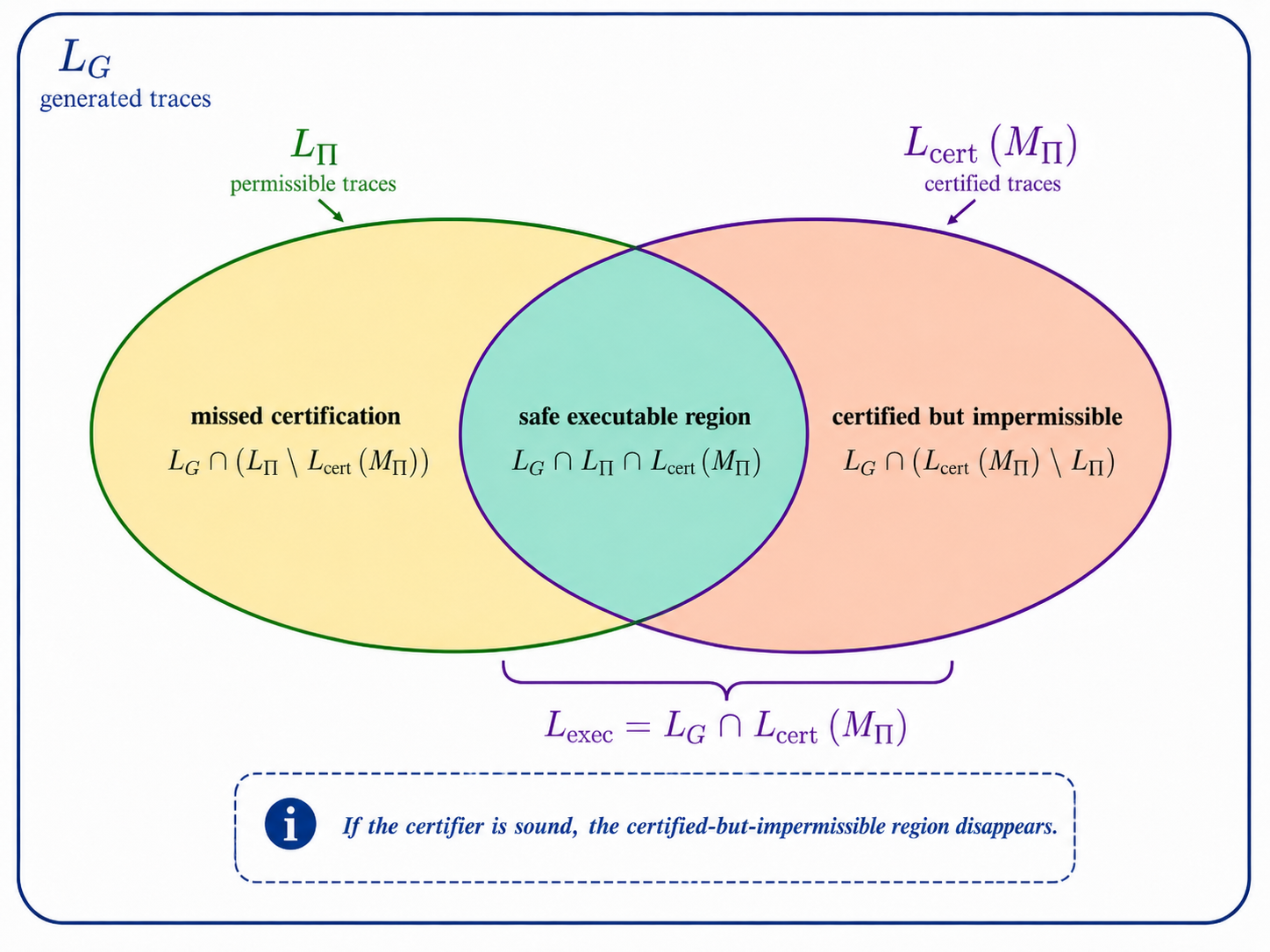}
    \caption{Certification boundary. Generated traces $L_G$ are partitioned by semantic permissibility
$L_\Pi$ and operational certification $L_{\mathrm{cert}}(M_\Pi)$. The executable region is
$L_{\mathrm{exec}}=L_G\cap L_{\mathrm{cert}}(M_\Pi)$; the two boundary errors are missed
certification, $L_G\cap(L_\Pi\setminus L_{\mathrm{cert}}(M_\Pi))$, and certified-but-impermissible
risk, $L_G\cap(L_{\mathrm{cert}}(M_\Pi)\setminus L_\Pi)$.}
\label{fig:certification-boundary}
    \vspace{-0.1in}
\end{figure}

Fig.~\ref{fig:certification-boundary} visualizes the trace-level certification boundary induced by
semantic permissibility and operational certification. The outer region $L_G$ denotes generated
traces. The green set $L_\Pi$ denotes traces that are semantically permissible under the policy
system, while the purple set $L_{\mathrm{cert}}(M_\Pi)$ denotes traces certified by the Permissibility
Machine. Their overlap inside $L_G$ is the safe executable region,
$L_G\cap L_\Pi\cap L_{\mathrm{cert}}(M_\Pi)$. The left-only region
$L_G\cap (L_\Pi\setminus L_{\mathrm{cert}}(M_\Pi))$ corresponds to missed certification:
permissible traces that the certifier fails to certify. The right-only region
$L_G\cap (L_{\mathrm{cert}}(M_\Pi)\setminus L_\Pi)$ corresponds to certified-but-impermissible
risk: traces certified by the system but not actually permissible. The executable region is
$L_{\mathrm{exec}}=L_G\cap L_{\mathrm{cert}}(M_\Pi)$; if the certifier is sound, then
$L_{\mathrm{cert}}(M_\Pi)\subseteq L_\Pi$, and the certified-but-impermissible region disappears.

\section{Privacy and Zero-Knowledge Trace Certification}
\label{sec:privacy-zk}

Privacy is one of the clearest examples of trace-dependent permissibility, connecting certified
traces to information-flow security and privacy-preserving verification~\citep{
sabelfeld2003informationflow}.  A final output may
contain no personally identifiable information (PII), while the trace that produced it may have
accessed, stored, logged, embedded, or leaked sensitive information. Thus privacy compliance
cannot be certified by output inspection alone:
\[
\text{Clean Output} \not\Rightarrow \tau\in L_{\Pi_{\mathrm{priv}}}.
\]

\textbf{Privacy as a trace property.}
Let $\Pi_{\mathrm{priv}}$ denote a privacy policy that may
constrain which user or agent may access which data, for what purpose, through which component,
and under which logging, retention, and release conditions. These obligations concern the information-flow trace: prompts, retrievals, tool calls,
memory updates, logs, caches, vector stores, and released outputs.

\textbf{Query budgets.}
A simple example is a query-budget policy. One query to a protected resource may be
authorized, while repeated queries within a time window may violate privacy or access-control
constraints. Let
\(
N_\tau(u,R,W)
\)
denote the number of queries by user or agent $u$ to protected resource $R$ within window $W$
along trace $\tau$. A privacy policy may require
\(
N_\tau(u,R,W)\le B.
\)
Each individual query may be permissible, while the full trace becomes impermissible once the
budget is exceeded. 

\textbf{Privacy--observability tension.}
Certification requires evidence about the trace, but the trace itself may contain sensitive data.
A certifier may need to know whether a protected source was accessed, whether PII flowed to an
unauthorized component, whether a query budget was exceeded, or whether an output used
approved data. Revealing the raw trace to prove these facts may itself violate privacy.

\textbf{Zero-knowledge certificates.}
This motivates privacy-preserving certification. A zero-knowledge certificate can prove selected
properties of an execution trace without revealing the sensitive trace itself. For example, a
certificate may establish that
\(
\tau\in L_{\Pi_{\mathrm{priv}}}
\)
without exposing the underlying prompts, private records, retrieved documents, model inputs, or
intermediate computations.

\textbf{Certifiable private trace properties.}
Privacy-relevant certificates may establish that all source accesses were authorized, no PII was
released to an unauthorized component, a query budget was not exceeded, a computation used
approved data without revealing the raw data, or an agent satisfied a policy predicate without
exposing private prompts or records.

\textbf{Interpretation.}
Privacy creates a fundamental tension for trustworthy agents: trace certification needs evidence,
but privacy limits evidence exposure. Zero-knowledge trace certificates address this tension by
separating the property being certified from the sensitive data used to certify it. This motivates
privacy-preserving proof memory as a core component of certified AI-agent execution.

\section{Limitations, Failure Modes, and Open Problems}
\label{sec:objections-agenda}

Certified traces define a formal execution-governance framework, not a completed verification stack.
We close by addressing alternative views, limitations, and open problems for trustworthy AI agents.

\paragraph{Objection 1: Stronger models will solve this.}
Stronger generative models may reduce factual errors, hallucinations, and tool-use mistakes, but
they do not by themselves create permission. Capability is not authority. A model may correctly
infer an action or retrieve a useful source while the trace violates access control, privacy, mandate,
evidence, or execution policy. 

\paragraph{Objection 2: Output guardrails are enough.}
Output guardrails are useful but incomplete. Many violations occur inside the trace: unauthorized
retrieval, leakage through prompts or logs, unreplayable computation, cumulative exposure,
purpose misuse, or policy-history violations. Trace certification asks a stronger question: not only
whether the output looks acceptable, but whether the process that produced it satisfies the policy
system.

\paragraph{Objection 3: Certification is too expensive.}
Full certification of every trace may be infeasible or unnecessary. A practical path is to certify
high-stakes traces first: actions that access sensitive data, modify records, generate regulated
recommendations, transfer assets, or trigger external execution. Certification can be made
incremental through reusable proof memory, partial certificates, domain-specific certifiers, and
escalation when evidence is incomplete. Human review is therefore a design option, not the definition of safety. The core requirement is
that high-risk execution be supported by checkable authorization evidence; such evidence
may include policy-engine checks, scoped credentials, cryptographic proofs, or human approvals.  Risk-tiered certification is the operational response to this objection: certificate burden should scale
with execution risk, irreversibility, regulatory exposure, privacy sensitivity, and systemic impact.


\textbf{Limitations.}
Certified traces do not solve value disagreement, policy design, institutional legitimacy, or
adversarial policy capture. A harmful, incomplete, or mis-specified policy can still certify harmful
behavior. Trace logging can also create privacy and surveillance risks if implemented naively.
Certification must therefore be paired with policy governance, privacy-preserving proof memory,
access control, and escalation. This also reinforces the need for independent evaluation and red teaming of agentic systems, an issue
emphasized by recent ICML position work on safe harbor for AI evaluation~\citep{longpre2024safeharbor}.

Table~\ref{tab:failure_modes} summarizes common failure modes for certified traces
and the corresponding design responses.

\begin{table}[t]
\centering
\small
\begin{tabular}{p{0.2\linewidth}p{0.32\linewidth}p{0.36\linewidth}}
\toprule
\textbf{Failure mode} & \textbf{Example} & \textbf{Certified-trace response} \\
\midrule
Fake certificate &
The agent fabricates an approval sentence or cites a nonexistent policy check. &
Certificates must be checkable proof objects, not natural-language self-attestations. \\

Stale policy &
A trace is certified under an old mandate, access rule, or regulatory policy. &
Certificates must bind to policy lineage: version, authority, timestamp, and effective date. \\

Stale approval &
A human approved an earlier trace, but the agent modifies the beneficiary, amount, source, or action. &
Approvals should bind to a trace hash, scope, time window, and execution conditions. \\

Tool deviation &
A tool returns a different result or executes a different call than the certified proposed trace. &
Realized-trace conformance must be checked; deviations halt, roll back, or escalate. \\

Hidden policy state &
The agent cannot see a restricted list, query budget, client mandate, or exposure limit. &
Proof memory and policy-state interfaces must expose the facts needed for certification. \\

Certificate laundering &
A weak certifier approves a high-risk trace outside its authority. &
Risk-tiered certification should specify which certifier class is authorized for each trace type. \\

Privacy leakage &
The raw trace needed for certification contains sensitive client data or private prompts. &
Use selective disclosure, commitments, access controls, or zero-knowledge trace certificates. \\
\bottomrule
\end{tabular}
\caption{Failure modes for certified traces and corresponding design responses.}
\label{tab:failure_modes}
\end{table}

\paragraph{Supervisory technology: Execution Certificates for Financial AI Agents.}
In regulated finance, certified traces can be instantiated as a supervisory technology. A
\emph{Financial AI Agent Execution Certificate} is a machine-readable record attached to a
consequential agent action, recommendation, report, payment instruction, trade proposal, or
compliance decision. Such a certificate may include:
\[
\pi_{\mathrm{fin}} =
(\pi_{\mathrm{agent}},
 \pi_{\mathrm{principal}},
 \pi_{\mathrm{purpose}},
 \pi_{\mathrm{source}},
 \pi_{\mathrm{policy}},
 \pi_{\mathrm{risk}},
 \pi_{\mathrm{approval}},
 \pi_{\mathrm{compute}},
 \pi_{\mathrm{privacy}},
 \pi_{\mathrm{conform}}).
\]
Here $\pi_{\mathrm{agent}}$ identifies the AI agent and system version; $\pi_{\mathrm{principal}}$
identifies the human, desk, client, fund, or institution on whose behalf the agent acts;
$\pi_{\mathrm{purpose}}$ records the business purpose and permitted scope; $\pi_{\mathrm{source}}$
records source access and evidence lineage; $\pi_{\mathrm{policy}}$ records applicable regulations,
firm policies, client mandates, and policy versions; $\pi_{\mathrm{risk}}$ records risk-limit,
restricted-list, exposure, liquidity, or payment-threshold checks; $\pi_{\mathrm{approval}}$ records
human approval or escalation when required; $\pi_{\mathrm{compute}}$ records replayable formulas,
code, units, periods, and computed values; $\pi_{\mathrm{privacy}}$ records privacy, retention, and
information-flow checks; and $\pi_{\mathrm{conform}}$ records whether the realized trace conformed
to the certified proposed trace.

This certificate is not intended to replace existing regulation. It is a technical object that can make
existing supervisory obligations auditable for agentic systems. In this sense, certified traces provide
a bridge from AI-agent architecture to financial supervision: regulators and firms can ask not only
what an agent output, but what certificate authorized the action.

\textbf{Research agenda.}
Certified traces suggest several directions for the machine learning community:
\begin{itemize}[leftmargin=*, label=$\bullet$]
    \item \textbf{Certifiable trace languages.} Which classes of agent traces admit efficient sound
    certification?
    \item \textbf{Proof-memory capacity and expressiveness.} What proof-relevant history is needed
    to certify cumulative, temporal, privacy, or information-flow policies?
    \item \textbf{Certification-boundary benchmarks.} How should agents be evaluated by generated
    impermissibility, certified-but-impermissible risk, missed certification, and executable yield?
    \item \textbf{Monitorability versus certification gaps.} What information is visible to a monitor,
    what is retained in proof memory, and what is sufficient to authorize execution?
    \item \textbf{Privacy-preserving trace certificates.} How can systems prove policy compliance
    without exposing sensitive traces, prompts, records, or computations?
    \item \textbf{Bounded certification.} What tradeoffs arise among proof memory, latency, privacy,
    verification cost, and certification error?
    \item \textbf{Policy drift and re-certification.} How should certificates be maintained when laws,
    mandates, access rules, or institutional policies change?
    \item \textbf{Compositional certification.} When do certificates for tools, agents, or subtasks
    compose into certificates for multi-agent traces?
\end{itemize}

\textbf{Conclusion.}
Trustworthy AI agents require more than generation. They require a theory of permissible
execution. The central object is not the model output alone, but the certificate-bearing trace:
a generator proposes, a Permissibility Machine certifies, and an executor acts only after
certification. No certificate, no execution.

\section*{Acknowledgments}

Xiao-Yang Liu Yanglet and Xiaodong Wang acknowledge the support from Columbia's SIRS and STAR Program, and The Tang Family Fund for Research Innovations in FinTech, Engineering, and Business Operations. Xiao-Yang Liu Yanglet acknowledges the support from NSF IUCRC CRAFT center research grant (CRAFT Grant 22017) for this research. The opinions expressed in this publication do not necessarily represent the views of NSF IUCRC CRAFT.

Xiao-Yang Liu Yanglet and Xiaodong Wang acknowledge the JPMorganChase Faculty
Research Award. Any views or opinions expressed herein are solely those of the authors
listed and may differ from the views and opinions expressed by JPMorganChase or its
affiliates. This material is not a product of the Research Department of J.P. Morgan
Securities LLC. This material should not be construed as an individual recommendation
for any particular client and is not intended as a recommendation of particular securities,
financial instruments, or strategies for any particular client. This material does not
constitute a solicitation or offer in any jurisdiction.

\bibliographystyle{unsrt}
\bibliography{ref, Refv2}

\appendix

\section{Formal Trace Model}
\label{app:formal-trace-model}

This appendix gives the formal trace model underlying the main text. Its purpose is to
make explicit the objects, languages, and assumptions used by the proposal--certification--execution architecture.

\paragraph{Trace universe.}
Let $\Sigma$ be an alphabet used to encode states, actions, tool calls, retrieval events, evidence
records, computation steps, memory updates, and execution events. A finite agent trace is represented
as a word
\[
\tau\in\Sigma^*.
\]
When the context is clear, we write a structured state--action trace as
\[
\tau=(s_0,a_1,s_1,\ldots,a_T,s_T),
\]
where $s_t\in S$ denotes the system state at time $t$, $a_t\in C$ denotes a candidate action, and
$T<\infty$ is the trace horizon.

\paragraph{Proposed and realized traces.}
We distinguish a proposed trace from a realized trace. A proposed trace
$\tau^{\mathrm{prop}}\in\Sigma^*$ is a structured pre-execution record of intended actions, tool
calls, evidence, approvals, credentials, computations, and execution conditions. This is the object
submitted to the Permissibility Machine for certification. A realized trace
$\tau^{\mathrm{real}}\in\Sigma^*$ records the events actually performed after execution. Certification
authorizes execution of the proposed trace, while auditing checks whether the realized trace conforms to
the certified proposed trace. For long-running workflows, certification may occur repeatedly on
trace prefixes before irreversible actions:
\[
\tau^{\mathrm{prop}}_{1:t}
\quad\longrightarrow\quad
M_\Pi(\tau^{\mathrm{prop}}_{1:t},\Pi)
\quad\longrightarrow\quad
\{\text{authorize, reject, or escalate}\}.
\]

\paragraph{Transition model.}
A trace may follow a transition rule
\[
s_t=\delta(s_{t-1},a_t),\qquad 1\le t\le T,
\]
where $\delta:S\times C\to S$ is an abstract transition map. The deterministic notation is used
only for simplicity. Randomness, tool outputs, external observations, and environmental events can
be encoded into the state.

\paragraph{Action-level objects.}
The computable action set $C$ denotes the candidate actions that the agent can represent, generate,
or submit. The action-level permissible set
\[
P\subseteq C
\]
contains actions that are permitted when considered in isolation. The action-level permissibility gap is
\[
C\setminus P.
\]
This action-level distinction is useful but insufficient: a trace may be impermissible even when all
of its individual actions lie in $P$.

Fig.~\ref{fig:action-level-geometry} gives the action-level geometry of the
separation principle. The set \(C\) is the operative universe of computable
candidate actions, while \(P\subsetneq C\) is the subset of actions permitted at
the action level. The region \(C\setminus P\) is the permissibility gap: actions
that may be generated or submitted by a computational system, but should not be
permitted for execution.

\begin{figure}[t]
    \centering
    \includegraphics[width=0.75\textwidth]{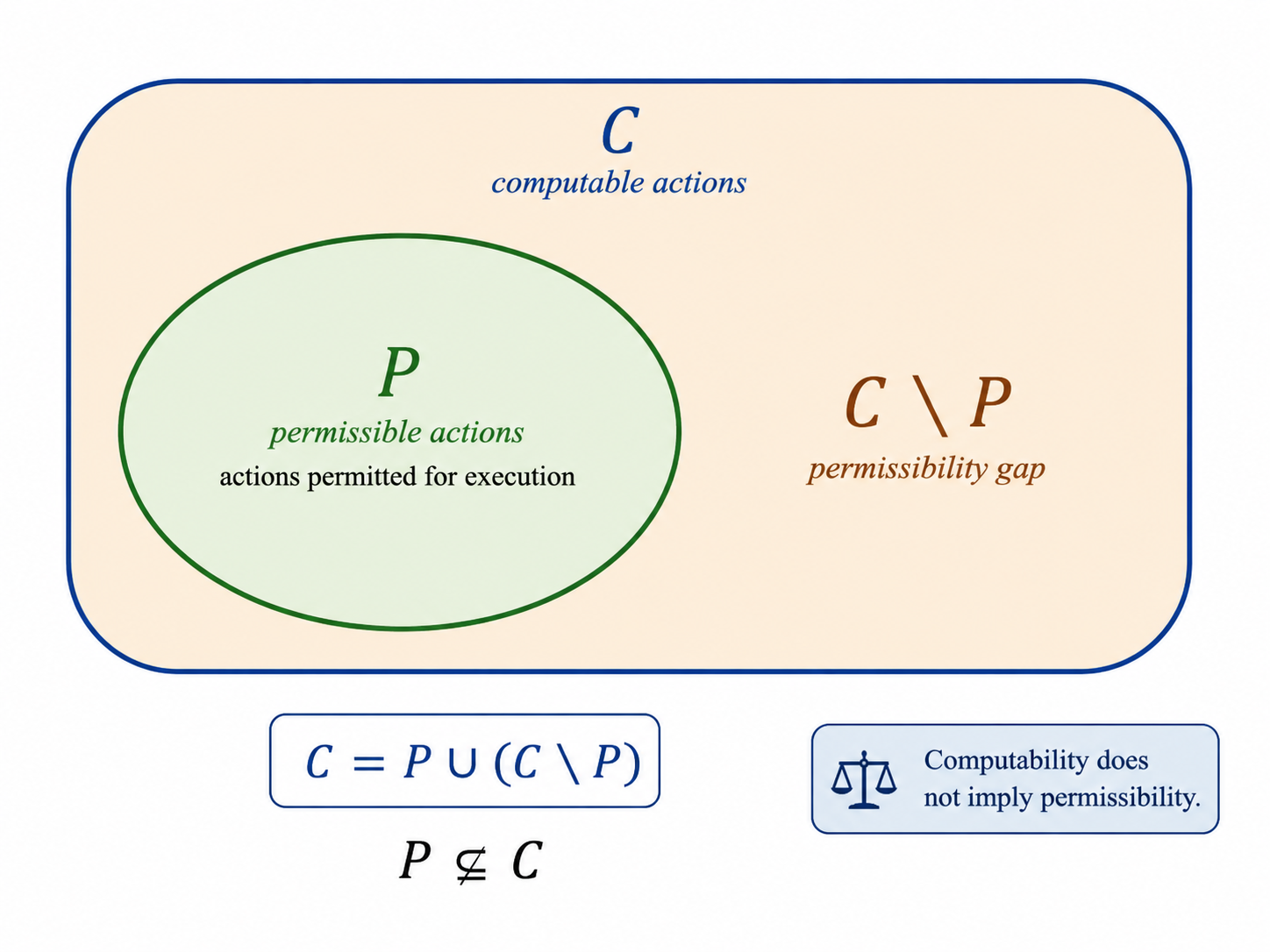} 
    \caption{
\textbf{Action-level geometry of the permissibility gap.}
The computable action set \(C\) is the operative universe of candidate actions
available to the system. The permissible action set \(P\subseteq C\) contains
actions permitted at the action level. The region \(C\setminus P\) is the
\emph{permissibility gap}: actions that are available to the system but not
permitted. In the strict-separation setting considered here, \(P\subsetneq C\),
so the gap is nonempty.
}
\label{fig:action-level-geometry}
\end{figure}

\paragraph{Trace-level languages.}
The generated trace language is
\[
L_G:=L(M_G)\subseteq \Sigma^*,
\]
where $M_G$ is the generating machine. The policy system $\Pi$ induces a semantic permissible
trace language
\[
L_\Pi\subseteq\Sigma^*.
\]
The certified trace language is
\[
L_{\mathrm{cert}}(M_\Pi)
=
\{\tau\in\Sigma^*:M_\Pi(\tau,\Pi)=\pi\},
\]
where $M_\Pi$ is the Permissibility Machine and $\pi$ is a checkable certificate. The executable
trace language is
\[
L_{\mathrm{exec}}
=
L_G\cap L_{\mathrm{cert}}(M_\Pi).
\]

\paragraph{Semantic versus operational regions.}
The language $L_\Pi$ is semantic: it represents the traces that are permissible under the policy
system. The language $L_{\mathrm{cert}}(M_\Pi)$ is operational: it represents the traces that the
certifier can certify. Sound certification requires
\[
L_{\mathrm{cert}}(M_\Pi)\subseteq L_\Pi.
\]
Completeness is not required. In realistic settings, some permissible traces may be rejected or
escalated because the certifier lacks sufficient proof memory, evidence, time, or authority.

\begin{table}[t]
\centering
\begin{tabular}{ll}
\hline\hline
Symbol & Meaning \\
\hline
$S$ & System state space \\
$C$ & Computable action set \\
$P$   & Permissible action set (action level) \\
$C\setminus P$ & Permissibility gap \\
$\Sigma$ & Alphabet of primitive trace symbols \\
$\tau$ & Finite agent trace \\
$\Pi$ & Policy system (permissibility calculus) \\
$\pi$ & Certificate \\
$M_G$ & Generating machine \\
$M_\Pi$ & Certifying machine (Permissibility Machine) \\
$E$ & External execution system \\
\hline
\end{tabular}
\caption{Notations for primitive objects.}
\label{tab:primitive-notation}
\end{table}

\section{Proofs of Main Results}
\label{app:proofs}

This appendix contains proofs omitted or abbreviated in the main text.

\subsection{Proof of Proposition~\ref{prop:trace-noncompositionality}}

\paragraph{Restatement.}
There exist policy systems and traces whose individual actions are locally permissible, while the
full trace is impermissible under the policy system.

\begin{proof}
We give a concrete construction using a query-budget policy. Let $R$ be a protected resource, such
as a password-related or identity-related database table. Let each action $a_i$ be a query to $R$ by
an authenticated user or agent $u$. Define the local action-level permissible set $P$ so that each
authenticated query is locally permissible:
\[
a_i\in P.
\]

Now define a trace-level privacy policy $\Pi_{\mathrm{priv}}$ that imposes a query budget. Let
\[
N_\tau(u,R,W)
\]
denote the number of queries by $u$ to resource $R$ within time window $W$ along trace $\tau$.
The policy requires
\[
N_\tau(u,R,W)\le 5.
\]
Construct a trace $\tau$ in which the same authenticated user or agent issues six queries to $R$
within the same time window $W$. Each individual query is locally permissible, so
\[
a_i\in P,\qquad i=1,\ldots,6.
\]
However,
\[
N_\tau(u,R,W)=6>5,
\]
so the full trace violates the privacy policy. Therefore
\[
\tau\notin L_{\Pi_{\mathrm{priv}}},
\]
although every individual action in the trace is locally permissible. Hence local action-level
permissibility need not compose into trace-level permissibility.
\end{proof}

\begin{table}[t]
\centering
\begin{tabular}{ll}
\hline\hline
Notation & Meaning \\
\hline
$L(\cdot)$  & Language generated/recognized by a machine \\
\(L_G := L(M_G)\) & Language generated by the generating machine \\
\(L_\Pi\) & Permissible trace language induced by policy system \(\Pi\) \\
\(L_{\mathrm{cert}}(M_\Pi)\) & Certified language recognized by the certifier \\
\(L_{\mathrm{exec}} := L_G \cap L_{\mathrm{cert}}(M_\Pi)\) & Executable language after certification \\
\(L_{\mathrm{task}}\) & Task language: traces that accomplish a specified task \\
\(\mathcal A=(M_G,M_\Pi,E)\) & Proposal-certification architecture with execution system \\
\(M_G^\Pi\) & Policy-conditioned generator \\
\(\Pi(M_G)\) & Runtime-enforced (sandboxed) generator \\
$\Pi \vdash \tau:\mathrm{Permitted}$ & 
Derivability judgment for permissibility \\
\(\Pi_\ell\) & Layer of a structured policy system \(\Pi\) \\
\(\Pi_{\le k}\) & Partial policy stack through layer \(k\) \\
\(\mathrm{Gap}_\Pi(M_G)\) & Generated permissibility risk \\
\hline
\(u_\Pi(M_G,M_\Pi)\) & Certified-but-impermissible risk of the certifier \\
\(m_\Pi(M_G,M_\Pi)\) & Missed-certification risk of the certifier \\
\(y_\Pi(M_G,M_\Pi)\) & Executable yield after certification \\
\(K_\Pi(\tau)\) & Certificate complexity or proof burden of trace \(\tau\) \\
\(D_\mu(t,t')\) & Policy-lineage drift between policy systems \(\Pi_t\) and \(\Pi_{t'}\) \\
\(\mathsf P_{\mathrm{perm}}, \mathsf{NP}_{\mathrm{perm}}\)
& Permissibility search--verification complexity classes \\
\hline
\end{tabular}
\caption{Derived languages and architecture-level notations.}
\label{tab:derived-notation}
\end{table}

\subsection{Proof of Theorem~\ref{thm:execution-safety}}

\paragraph{Restatement.}
Under sound certification and execution authorization, every executed trace is permissible:
\[
\mathrm{Execute}(\tau)\Longrightarrow \tau\in L_\Pi.
\]

\begin{proof}
By execution authorization,
\[
\mathrm{Execute}(\tau)
\Longrightarrow
\tau\in L_{\mathrm{cert}}(M_\Pi).
\]
By sound certification,
\[
L_{\mathrm{cert}}(M_\Pi)\subseteq L_\Pi.
\]
Therefore $\tau\in L_\Pi$.
\end{proof}

\subsection{Proof of Theorem~\ref{thm:exec-language-composition}}

\paragraph{Restatement.}
For a fixed policy system $\Pi$ and architecture
\[
\mathcal A=(M_G,M_\Pi,E),
\]
a trace is executable if and only if it is both generated and certified. Equivalently,
\[
L_{\mathrm{exec}}(\mathcal A)
=
L_G\cap L_{\mathrm{cert}}(M_\Pi).
\]

\begin{proof}
If a trace is executable under $\mathcal A$, then it must have been proposed by the generating
machine. Hence $\tau\in L_G$. By execution authorization, the trace must also be certified before
execution, so $\tau\in L_{\mathrm{cert}}(M_\Pi)$. Therefore
\[
\tau\in L_G\cap L_{\mathrm{cert}}(M_\Pi).
\]

Conversely, suppose
\[
\tau\in L_G\cap L_{\mathrm{cert}}(M_\Pi).
\]
Then $\tau$ is a generated trace and certification succeeds. Under the architecture, generated traces
are authorized for execution only after certification succeeds. Hence $\tau$ is executable under
$\mathcal A$. Therefore
\[
L_{\mathrm{exec}}(\mathcal A)
=
L_G\cap L_{\mathrm{cert}}(M_\Pi).
\]
\end{proof}

\subsection{Proof of Theorem~\ref{thm:boundary-decomposition}}

Recall the definitions
\[
\mathrm{Gap}_\Pi(M_G)
=
\Pr_{\tau\sim\mu_G}[\tau\notin L_\Pi],
\]
\[
u_\Pi(M_G,M_\Pi)
=
\Pr_{\tau\sim\mu_G}
[
\tau\in L_{\mathrm{cert}}(M_\Pi)\setminus L_\Pi
],
\]
\[
m_\Pi(M_G,M_\Pi)
=
\Pr_{\tau\sim\mu_G}
[
\tau\in L_\Pi\setminus L_{\mathrm{cert}}(M_\Pi)
],
\]
and
\[
y_\Pi(M_G,M_\Pi)
=
\Pr_{\tau\sim\mu_G}
[
\tau\in L_{\mathrm{cert}}(M_\Pi)
].
\]

\paragraph{Restatement.}
For any policy system $\Pi$, generator $M_G$, and certifier $M_\Pi$,
\[
y_\Pi
=
1-\mathrm{Gap}_\Pi-m_\Pi+u_\Pi.
\]
If the certifier is sound, then $u_\Pi=0$ and
\[
y_\Pi
=
1-\mathrm{Gap}_\Pi-m_\Pi.
\]

\begin{proof}
For readability, we write
\[
L_{\mathrm{cert}}:=L_{\mathrm{cert}}(M_\Pi).
\]
The certified language decomposes into its permissible and impermissible parts:
\[
L_{\mathrm{cert}}
=
(L_{\mathrm{cert}}\cap L_\Pi)
\cup
(L_{\mathrm{cert}}\setminus L_\Pi),
\]
and the union is disjoint. Therefore
\[
y_\Pi
=
\Pr_{\tau\sim\mu_G}[\tau\in L_{\mathrm{cert}}\cap L_\Pi]
+
\Pr_{\tau\sim\mu_G}[\tau\in L_{\mathrm{cert}}\setminus L_\Pi].
\]
The second term is $u_\Pi$. For the first term,
\[
\Pr_{\tau\sim\mu_G}[\tau\in L_{\mathrm{cert}}\cap L_\Pi]
=
\Pr_{\tau\sim\mu_G}[\tau\in L_\Pi]
-
\Pr_{\tau\sim\mu_G}[\tau\in L_\Pi\setminus L_{\mathrm{cert}}].
\]
By definition,
\[
\Pr_{\tau\sim\mu_G}[\tau\in L_\Pi]
=
1-\mathrm{Gap}_\Pi,
\]
and
\[
\Pr_{\tau\sim\mu_G}[\tau\in L_\Pi\setminus L_{\mathrm{cert}}]
=
m_\Pi.
\]
Combining these identities gives
\[
y_\Pi
=
1-\mathrm{Gap}_\Pi-m_\Pi+u_\Pi.
\]
If the certifier is sound, then
\[
L_{\mathrm{cert}}(M_\Pi)\subseteq L_\Pi,
\]
so
\[
L_{\mathrm{cert}}(M_\Pi)\setminus L_\Pi=\emptyset
\]
and hence $u_\Pi=0$. The sound-certifier identity follows.
\end{proof}

\subsection{Residual Execution Risk}

\paragraph{Definition.}
Assume $y_\Pi(M_G,M_\Pi)>0$. The residual execution risk is the probability that an executed
trace is impermissible:
\[
\rho_\Pi(M_G,M_\Pi)
=
\Pr_{\tau\sim\mu_G}
[
\tau\notin L_\Pi
\mid
\tau\in L_{\mathrm{cert}}(M_\Pi)
].
\]

\paragraph{Claim.}
If $y_\Pi(M_G,M_\Pi)>0$, then
\[
\rho_\Pi(M_G,M_\Pi)
=
\frac{u_\Pi(M_G,M_\Pi)}{y_\Pi(M_G,M_\Pi)}.
\]
In particular, a sound certifier has zero residual execution risk.

\begin{proof}
By conditional probability,
\[
\Pr[\tau\notin L_\Pi\mid \tau\in L_{\mathrm{cert}}(M_\Pi)]
=
\frac{
\Pr[\tau\in L_{\mathrm{cert}}(M_\Pi)\setminus L_\Pi]
}{
\Pr[\tau\in L_{\mathrm{cert}}(M_\Pi)]
}.
\]
The numerator is $u_\Pi(M_G,M_\Pi)$ and the denominator is
$y_\Pi(M_G,M_\Pi)$. If the certifier is sound, then
\[
L_{\mathrm{cert}}(M_\Pi)\subseteq L_\Pi,
\]
so $u_\Pi(M_G,M_\Pi)=0$ and hence $\rho_\Pi(M_G,M_\Pi)=0$.
\end{proof}

\subsection{Permissibility Recall}

\paragraph{Definition.}
Assume $1-\mathrm{Gap}_\Pi(M_G)>0$. The permissibility recall is the probability that a permissible
generated trace is certified:
\[
\mathrm{Recall}_\Pi(M_G,M_\Pi)
=
\Pr_{\tau\sim\mu_G}
[
\tau\in L_{\mathrm{cert}}(M_\Pi)
\mid
\tau\in L_\Pi
].
\]

\paragraph{Claim.}
If $1-\mathrm{Gap}_\Pi(M_G)>0$, then
\[
\mathrm{Recall}_\Pi(M_G,M_\Pi)
=
1-
\frac{m_\Pi(M_G,M_\Pi)}{1-\mathrm{Gap}_\Pi(M_G)}.
\]

\begin{proof}
By conditional probability,
\[
\Pr[\tau\in L_{\mathrm{cert}}(M_\Pi)\mid \tau\in L_\Pi]
=
\frac{
\Pr[\tau\in L_{\mathrm{cert}}(M_\Pi)\cap L_\Pi]
}{
\Pr[\tau\in L_\Pi]
}.
\]
The denominator is $1-\mathrm{Gap}_\Pi(M_G)$. The numerator is
\[
\Pr[\tau\in L_\Pi]
-
\Pr[\tau\in L_\Pi\setminus L_{\mathrm{cert}}(M_\Pi)]
=
1-\mathrm{Gap}_\Pi(M_G)-m_\Pi(M_G,M_\Pi).
\]
Dividing gives the result.
\end{proof}

\subsection{Monotonicity Under Conjunctive Policy Strengthening}

\paragraph{Claim.}
Let
\[
\Pi'=\Pi\wedge \Pi_{\mathrm{new}},
\]
where $\Pi_{\mathrm{new}}$ adds an obligation conjunctively and does not introduce an exception or
override that enlarges the permissible language. Then
\[
L_{\Pi'}\subseteq L_\Pi.
\]
Consequently, for a fixed generator distribution $\mu_G$,
\[
\mathrm{Gap}_{\Pi'}(M_G)\ge \mathrm{Gap}_{\Pi}(M_G).
\]

\begin{proof}
A trace belongs to $L_{\Pi'}$ only if it satisfies both $\Pi$ and $\Pi_{\mathrm{new}}$. Therefore
every trace permissible under $\Pi'$ is permissible under $\Pi$, so
\[
L_{\Pi'}\subseteq L_\Pi.
\]
Taking complements in the common trace universe $\Sigma^*$ reverses the inclusion:
\[
\Sigma^*\setminus L_\Pi
\subseteq
\Sigma^*\setminus L_{\Pi'}.
\]
Equivalently,
\[
\{\tau\in\Sigma^*:\tau\notin L_\Pi\}
\subseteq
\{\tau\in\Sigma^*:\tau\notin L_{\Pi'}\}.
\]
Taking $\mu_G$-probabilities gives
\[
\Pr_{\tau\sim\mu_G}[\tau\notin L_\Pi]
\le
\Pr_{\tau\sim\mu_G}[\tau\notin L_{\Pi'}].
\]
Thus
\[
\mathrm{Gap}_{\Pi'}(M_G)\ge \mathrm{Gap}_{\Pi}(M_G).
\]
\end{proof}

\paragraph{Caveat.}
This monotonicity statement assumes conjunctive strengthening. It need not hold if the new policy
layer introduces exceptions, overrides, or relaxations that enlarge the permissible language.

\section{Canonical Policy Models}
\label{app:canonical-policy-models}

This appendix develops several mathematical models of policy systems. These models are not
intended to be exhaustive. Their purpose is to show how common policy structures induce different
trace-certification problems.

\subsection{Semantic Trace-Language Policy}

The most abstract model represents a policy system only by its induced permissible trace language:
\[
\Pi \equiv L_\Pi\subseteq\Sigma^*.
\]
A trace is permissible when
\[
\tau\in L_\Pi.
\]
This model is sufficient for the axioms of certifiable execution, executable-language composition,
and certification-boundary measures. It does not require a specific syntax for policies.

\subsection{Automaton or Monitor Policy}

Some policies can be represented by finite-state monitors. A monitor policy has the form
\[
\Pi=(Q,q_0,\delta_\Pi,F),
\]
where $Q$ is a finite set of monitor states, $q_0$ is the initial state, $\delta_\Pi:Q\times\Sigma\to Q$
is a transition function, and $F\subseteq Q$ is the set of accepting states. The induced permissible
language is
\[
L_\Pi
=
\{\tau\in\Sigma^*:\delta_\Pi(q_0,\tau)\in F\}.
\]
This model captures workflow constraints, authentication states, local rule checks, and finite-state
runtime monitors.

\subsection{Counter or Resource-Budget Policy}

Counter policies constrain accumulated quantities along a trace. Let $z_t$ be a vector of counters
updated by
\[
z_t=z_{t-1}+g(s_{t-1},a_t,s_t),
\]
where $g$ records resource usage, exposure changes, query counts, or other cumulative quantities.
A counter policy requires
\[
z_t\in\mathcal Z_{\mathrm{safe}}
\qquad\text{for all }t.
\]
Examples include exposure limits, query budgets, rate limits, liquidity constraints, and daily trading
limits. These policies are trace-dependent because an action that is locally permissible may become
impermissible after accumulated usage crosses a threshold.

\subsection{Temporal or History-Dependent Policy}

Temporal policies impose ordering or timing obligations on traces. They may be represented by
finite-trace temporal formulas, such as formulas in $\mathrm{LTL}_f$ or related logics. A trace is
permissible when
\[
\tau\models \varphi_\Pi.
\]
Examples include approval-before-execution, deletion within a retention window, escalation after
unresolved evidence, prohibition of execution after a risk violation, and restrictions on reusing data
for a new purpose. Temporal policies make explicit that permissibility depends on event order and
history.

\subsection{Information-Flow and Access-Control Policy}

Information-flow policies govern which agents or components may access, transform, store, or
release which data. A trace may contain events
\[
e_t=(u_t,c_t,d_t,p_t,r_t),
\]
where $u_t$ is a user or agent, $c_t$ is a component or tool, $d_t$ is a data item or data class,
$p_t$ is a declared purpose, and $r_t$ is a retrieved, released, or stored record. A policy may require
\[
\mathrm{Auth}_\Pi(u_t,c_t,d_t,p_t)=1
\]
and may also constrain downstream flow, logging, retention, redaction, and release. This model is
important for PII, source authorization, purpose limitation, and privacy-preserving agent traces.

\subsection{Proof-Calculus Policy}

For proof-carrying execution, a policy system may be viewed as a calculus supporting judgments
\[
\Pi\vdash \tau:\mathrm{Permitted}.
\]
A certificate $\pi$ is a checkable witness for such a judgment. Soundness requires
\[
\Pi\vdash \tau:\mathrm{Permitted}
\Longrightarrow
\tau\in L_\Pi.
\]
This model distinguishes semantic permissibility from derivable permissibility. A trace may be
semantically permissible but not certifiable by the current proof system or proof memory.

\subsection{Versioned or Dynamic Policy}

In deployed systems, policies evolve. We write
\[
\Pi_t\longrightarrow \Pi_{t'}
\]
to denote a policy update from time $t$ to time $t'$. A certificate issued under $\Pi_t$ may not
remain valid under $\Pi_{t'}$. Versioned policies motivate policy lineage: a certificate should record
the policy version, source, authority, and effective date under which a trace was certified.

\paragraph{Summary.}
The abstract semantic model defines the target permissible region. Automaton, counter, temporal,
and information-flow models expose common policy structures. Proof-calculus policies support
certificates. Versioned policies explain why certification must be tied to policy lineage.

\section{Proof-Memory Models}
\label{app:proof-memory-models}

This appendix defines proof-memory models and distinguishes expressiveness from capacity.

\paragraph{Proof-memory interface.}
A proof-memory model is an abstract stateful structure
\[
\mathcal H=(\mathcal S,h_0,\mathrm{Update},\mathrm{Query},\mathrm{Emit}),
\]
where $\mathcal S$ is the memory-state space, $h_0\in\mathcal S$ is the initial memory state,
$\mathrm{Update}$ incorporates trace events, $\mathrm{Query}$ retrieves proof-relevant facts, and
$\mathrm{Emit}$ produces certificate components.

Given trace event $e_t$, the memory state evolves as
\[
h_t=\mathrm{Update}(h_{t-1},e_t).
\]
The memory state $h_t$ is not merely storage. It is the proof-relevant history available to the
certifier at time $t$.

\paragraph{Proof-memory expressiveness.}
Proof-memory expressiveness describes which trace properties or certified trace languages can be
supported by a given memory model. For a fixed policy system $\Pi$, let
\[
\mathsf{CertLang}_\Pi(\mathcal H)
=
\left\{
L_{\mathrm{cert}}(M_\Pi):
M_\Pi \text{ uses proof memory } \mathcal H
\text{ and }
L_{\mathrm{cert}}(M_\Pi)\subseteq L_\Pi
\right\}.
\]
We say that $\mathcal H_1$ is no more certification-expressive than $\mathcal H_2$, written
\[
\mathcal H_1\preceq_{\mathrm{cert}}\mathcal H_2,
\]
if
\[
\mathsf{CertLang}_\Pi(\mathcal H_1)
\subseteq
\mathsf{CertLang}_\Pi(\mathcal H_2).
\]
Thus expressiveness is qualitative: it concerns which trace obligations can be certified.

\paragraph{Proof-memory capacity.}
Proof-memory capacity is quantitative. For horizon $T$, let $\mathcal S_T$ be the set of memory
states reachable after traces of length at most $T$. A simple capacity measure is
\[
\mathrm{Cap}_{\mathcal H}(T)=\log |\mathcal S_T|.
\]
This measures how many bits of proof-relevant history the memory model can distinguish at
horizon $T$. A policy that requires distinguishing many history classes requires correspondingly
large proof-memory capacity.

\paragraph{Examples of proof-memory models.}
Common proof-memory models include:
\begin{itemize}
    \item \textbf{Finite-state memory:} supports local checks and simple workflow constraints.
    \item \textbf{Counter memory:} supports budgets, exposure limits, rate limits, and query counts.
    \item \textbf{Stack memory:} supports nested obligations, scoped approvals, and recursive tasks.
    \item \textbf{Key-value or table memory:} supports entity records, source permissions, and
    account-specific constraints.
    \item \textbf{DAG or provenance memory:} supports evidence binding, computation replay, and
    derivation traces.
    \item \textbf{Append-only ledger memory:} supports auditability, policy lineage, and durable
    certification records.
    \item \textbf{Cryptographic commitment or zero-knowledge memory:} supports privacy-preserving
    certification and selective disclosure.
\end{itemize}

\paragraph{Principle.}
A certifier can only certify trace properties that its proof memory can distinguish. This principle is
the memory analogue of trace non-compositionality: if a policy depends on history, then certification
requires access to the relevant history.

\section{Additional Boundary Measures}
\label{app:additional-boundary-measures}

This appendix collects additional measures for studying certification boundaries.

\paragraph{Certifier-policy distance.}
The distance between the operational certified language and the semantic permissible language is
\[
d_\mu(M_\Pi,\Pi)
=
\mu_G\!\left(L_{\mathrm{cert}}(M_\Pi)\triangle L_\Pi\right),
\]
where $\triangle$ denotes symmetric difference. Expanding the symmetric difference gives
\[
d_\mu(M_\Pi,\Pi)
=
u_\Pi(M_G,M_\Pi)+m_\Pi(M_G,M_\Pi).
\]
Thus certification error decomposes into certified-but-impermissible risk and missed-certification
risk.

\paragraph{Monitorability gap and certification gap.}
Let $O(\tau)$ denote the observable part of a trace available to a monitor. This may include final
outputs, actions, tool logs, activations, or chain-of-thought. A monitorability gap arises when
$O(\tau)$ lacks enough policy-relevant information to determine whether the trace is unsafe or
impermissible. A certification gap arises when the available observations and proof memory are
insufficient to construct a checkable certificate, even when the trace is semantically permissible.
This distinction separates detection from authorization: monitorability supports suspicion or
warning, while certifiability supports execution permission.

\paragraph{Residual execution risk.}
The residual execution risk is
\[
\rho_\Pi(M_G,M_\Pi)
=
\Pr_{\tau\sim\mu_G}
[
\tau\notin L_\Pi
\mid
\tau\in L_{\mathrm{cert}}(M_\Pi)
].
\]
When $y_\Pi(M_G,M_\Pi)>0$,
\[
\rho_\Pi(M_G,M_\Pi)
=
\frac{u_\Pi(M_G,M_\Pi)}{y_\Pi(M_G,M_\Pi)}.
\]
Sound certifiers have zero residual execution risk.

\paragraph{Permissibility recall.}
Permissibility recall is
\[
\mathrm{Recall}_\Pi(M_G,M_\Pi)
=
\Pr_{\tau\sim\mu_G}
[
\tau\in L_{\mathrm{cert}}(M_\Pi)
\mid
\tau\in L_\Pi
].
\]
When $1-\mathrm{Gap}_\Pi(M_G)>0$,
\[
\mathrm{Recall}_\Pi(M_G,M_\Pi)
=
1-
\frac{m_\Pi(M_G,M_\Pi)}{1-\mathrm{Gap}_\Pi(M_G)}.
\]

\paragraph{Certificate complexity.}
For a trace $\tau$, define its certificate complexity under policy system $\Pi$ by
\[
K_\Pi(\tau)
=
\min\{|\pi|:\Pi\vdash \tau:\mathrm{Permitted}\},
\]
with $K_\Pi(\tau)=\infty$ if no such certificate exists. This measures the proof burden needed to
move a trace across the certification boundary.

\paragraph{Expected certificate burden.}
For a generator distribution $\mu_G$, define
\[
\bar K_\Pi(M_G)
=
\mathbb E_{\tau\sim\mu_G}
[
K_\Pi(\tau)
\mid
\tau\in L_\Pi
],
\]
when the conditional expectation is well-defined. This measures the average proof burden among
permissible generated traces.

\paragraph{Policy-lineage drift.}
For two policy systems $\Pi_t$ and $\Pi_{t'}$, define a semantic policy distance
\[
d_\mu(\Pi_t,\Pi_{t'})
=
\mu\!\left(L_{\Pi_t}\triangle L_{\Pi_{t'}}\right).
\]
A deployment-specific drift measure is
\[
D_\mu(t,t')
=
\mu
\left(
L_{\mathrm{cert}}(M_{\Pi_t})
\setminus
L_{\Pi_{t'}}
\right).
\]
This measures the mass of traces certified under policy $\Pi_t$ that would no longer be permissible
under policy $\Pi_{t'}$.

\paragraph{Certification rate-distortion.}
Let $\mathcal M(B)$ denote a class of certifiers whose proof-memory capacity, verification latency,
or evidence budget is at most $B$. A generic certification distortion is
\[
D_\Pi(B)
=
\inf_{M_\Pi\in\mathcal M(B)}
\left[
\alpha u_\Pi(M_G,M_\Pi)
+
\beta m_\Pi(M_G,M_\Pi)
\right],
\]
where $\alpha,\beta\ge 0$ weight unsafe certification and missed certification. This quantity asks:
under a bounded certification resource, how well can the operational certified boundary approximate
the semantic permissible boundary?


\section{Illustrative Examples and Proof Memory}
\label{app:illustrative-examples}

The examples here explain why proof memory is a natural organizing
dimension for certifiers. In classical automata theory, memory controls what
languages can be recognized. For permissibility certification, memory controls
what trace properties can be certified. Some constraints are local; others
depend on prior actions, cumulative exposure, evidence provenance, computation
replay, or audit history. Thus proof memory is what allows a certifier to move
from local action checking to trace-level certification.

The following examples serve four purposes. An action
\(
a\in C
\)
may represent a trade, a portfolio decision, or an entire trading strategy, depending on abstraction level. 

First, the examples witness the nonemptiness of the permissibility gap
\(
C\setminus P \neq \varnothing.
\)
Second, they suggest why permissibility may require external certification, rather than generation constraints alone.
Third, this motivates organizing certifiers by the richness of the proof memory they
maintain:
\[
M_\Pi^{(0)}
\preceq
M_\Pi^{(1)}
\preceq
M_\Pi^{(2)}
\preceq
M_\Pi^{(3)}.
\]
where $\preceq$ denotes ``no more expressive than”. 
Richer proof memory may permit certification of richer permissible languages. Fourth, they show that action-level permissibility may not compose into trace-level permissibility.

The role of memory is analogous to state augmentation in sequential decision-making.
In a Markovian model, the current state must contain enough information from the past
to determine future transitions and evaluations.  Similarly, in permissibility
certification, proof memory records the information from the past needed to certify
future actions: prior trades, current holdings, exposures, evidence, computations,
approvals, and unresolved obligations. Thus proof memory may be viewed as a
certification-oriented state variable. Trace dependence does not contradict Markovian modeling;
it says that the state must contain the proof-relevant history.

\subsection{Computable but Impermissible Strategy}
\label{app:example-impermissible-strategy}

Let
\(
a\in C
\)
be a computable trading strategy exploiting a market microstructure pattern. Suppose
\(
a\notin P
\)
because it violates exchange conduct rules.
Then
\(
a\in C\setminus P.
\)

This motivates finite-state certifiers of class
\[
M_\Pi^{(0)}.
\]

\textit{Concrete instantiation}. Examples include spoofing-like order placement,
latency-sensitive stale-quote strategies,
or microstructure exploits.

This illustrates
\[
\text{Profitability}
\not\Rightarrow
\text{Permissibility}.
\]

\subsection{Optimality Without Permissibility}
\label{app:example-optimality}

Let
\(
w^*
\)
be a portfolio allocation maximizing return subject to conventional risk constraints. Suppose
\(
w^* \notin P
\)
because it induces hidden leverage concentration forbidden by regulation.
Then
\(
w^* \in C \setminus P.
\)
The optimization is valid,
yet the resulting action is impermissible.

This motivates certifiers with arithmetic proof memory of class
\[
M_{\Pi}^{(1)}.
\]

\emph{Concrete instantiation.}
Examples include leveraged portfolios that satisfy optimization objectives
yet violate concentration, liquidity, or exposure constraints. Here permissibility depends not only on feasibility,
but on proving structural constraints.

This illustrates
\[
\textit{Optimality}
\not\Rightarrow
\textit{Permissibility}.
\]
Thus richer certification may require more than finite-state recognition.

\subsection{Permissible Actions, Impermissible Trace}
\label{app:example-noncompositionality}

Suppose each trade \(a_i\) in a sequence passes local action-level checks:
\[
a_i\in P
\quad \text{for all } i.
\]
In financial applications, the policy system \(\Pi\) should be understood
schematically as combining multiple sources of constraints:
\[
\Pi
=
\text{law/regulation}
+
\text{fund mandate}
+
\text{client policy}
+
\text{internal risk rule}
+
\text{execution-control rule}.
\]
This expression is schematic rather than algebraic: it indicates that financial
permissibility is usually induced by a layered policy system. Section~3.7
develops this point further by connecting policy layers to proof obligations and
persistent proof memory.

For example, a portfolio policy may impose a concentration limit:
semiconductor exposure must remain below \(20\%\). Each individual buy order may
pass local checks: the security is tradable, the account is authorized, and the
single-order size is below a trading limit. Yet the sequence of orders may push
aggregate semiconductor exposure above \(20\%\). The violation is not visible
from the action name alone; it is visible only from the trace and the evolving
portfolio state.

Thus the sequence of trades may create an impermissible trace.  The following proposition formalizes this phenomenon by an explicit cumulative-exposure construction.

\begin{proposition}[Non-Compositionality of Action-Level Permissibility]
Action-level permissibility need not compose into trace-level permissibility.
There may exist a trace
\[
\tau=(s_0,a_1,s_1,\ldots,a_T,s_T)
\]
such that
\[
a_i\in P
\quad \text{for all } i,
\]
but
\[
\tau\notin L_\Pi.
\]
\end{proposition}

\begin{proof}
We prove the claim by construction. Consider a simplified portfolio state
$s_t=e_t \in [0,1]$, where $e_t$ denotes the portfolio exposure to a given sector
after time $t$. Let the computable actions be buy orders
\[
a(q) \in C, \qquad q \ge 0,
\]
where $q$ is the additional sector exposure created by the order. The transition rule is
\[
e_t = e_{t-1}+q_t
\]
when action $a(q_t)$ is taken.

Define action-level permissibility by a single-order limit:
\[
P=\{a(q)\in C: 0\le q \le 0.10\}.
\]
Thus every order of size at most $10\%$ sector exposure is permissible at the action level.

Now let the trace-level policy system $\Pi$ impose a cumulative exposure constraint:
\[
L_\Pi
=
\{\tau=(s_0,a_1,s_1,\ldots,a_T,s_T): e_t \le 0.20
\text{ for every } t\}.
\]
This says that the sector exposure must never exceed $20\%$ along the trace.

Take initial exposure $e_0=0$ and consider the trace generated by three buy orders
\[
a_1=a(0.10), \qquad a_2=a(0.10), \qquad a_3=a(0.10).
\]
Each action is individually permissible, since
\[
a_i\in P \qquad \text{for } i=1,2,3.
\]
However, the resulting exposure sequence is
\[
e_1=0.10,\qquad e_2=0.20,\qquad e_3=0.30.
\]
Since $e_3>0.20$, the trace violates the cumulative exposure constraint. Therefore
\[
\tau \notin L_\Pi,
\]
although every action in the trace belongs to $P$. Hence action-level permissibility need not
compose into trace-level permissibility.
\end{proof}

This illustrates
\[
\text{Permissible Actions}
\not\Rightarrow
\text{Permissible Trace}.
\]

\subsection{Correct Answer, Impermissible Trace}
\label{app:example-correct-answer-impermissible-trace}

Let
\(
\tau
\)
be a reasoning trace produced by a language model.
Suppose the resulting financial conclusion is correct,
but
\(
\tau
\)
contains fabricated citations, unsupported claims,
or unreplayable calculations.
Then
\(
\tau \notin L_\Pi.
\)
Correctness of answer does not imply permissibility of trace.

A financial reporting example makes this distinction concrete. Suppose an agent
writes that a firm's gross margin improved year over year. The statement may be
true as a final claim, but the trace may still be impermissible if the agent used
inconsistent periods, unsupported sources, incorrect XBRL tags, incompatible
units, or unreplayable arithmetic. A certifiable explanation must bind the claim
to a derivation trace: source filings, extracted facts, period alignment, unit
normalization, formula, computation, and claim generation.

This phenomenon appears concretely in XBRL-based financial report analysis,
where LLMs must bind claims to filing facts, taxonomy concepts, periods, units,
and computations~\citep{han2024xbrlagent}.

This motivates proof-carrying certifiers of class
\[
M_{\Pi}^{(2)}.
\]

\emph{Concrete instantiation.}
A financial agent may correctly infer leverage has declined,
yet cite nonexistent filings
or use inconsistent accounting periods.

This illustrates

\[
\textit{Correct Answer}
\not\Rightarrow
\textit{Permissible Trace}.
\]

\subsection{Escalation as a Permissible Action}
\label{app:example-escalation}

Suppose evidence is incomplete,
so neither approval nor rejection can be certified.
Then
\[
M_\Pi(\tau)= \mathrm{Escalate}.
\]

\begin{proposition}[Escalation Under Uncertainty]
There exist states in which
\[
\mathrm{Escalate}
\]
is the unique permissible action.
\end{proposition}

\begin{proof}
Suppose neither approval nor rejection is derivable under
\(
\Pi.
\)
Then execution is unjustified,
and rejection is unsupported.
Escalation is therefore the only permissible action. 
\end{proof}

In this sense, escalation is not merely a rejection alternative.
It may be the required continuation of certification under uncertainty.

This motivates escalation-aware certifiers of class
\[
M_\Pi^{(3)}.
\]

\emph{Concrete instantiation.}
In shell-company screening,
uncertain ownership structures may force escalation.

This illustrates
\[
\textit{High Confidence}
\not\Rightarrow
\textit{Permission}.
\]

\subsection{Proof-Carrying Execution}
\label{app:example-proof-carrying-execution}

Suppose a trading agent proposes the trace
\[
\tau=\text{``buy asset A''}.
\]
Execution requires a certificate
\[
\pi
=
(\pi_{\mathrm{suitability}},
\pi_{\mathrm{risk}},
\pi_{\mathrm{liquidity}}).
\]
Execution proceeds only when the trace \(\tau\) is accompanied by
a valid certificate \(\pi\).  No certificate, no execution.

This motivates certifiers for structured proof objects.

\emph{Concrete instantiation.}
A trade may be computationally feasible
yet executable only when accompanied by certifiable evidence.

This illustrates
\[
\textit{Provable Permissibility} \Rightarrow \textit{Executable Trace}.
\]
This elevates certification from filter to precondition for execution. This operationalizes the second purpose above:
permissibility may require certification rather than generation constraints alone.

\subsection{Persistent Proof Memory}
\label{app:example-persistent-proof-memory}

Some certifications require persistent proof memory.  A certifier may maintain
an append-only memory state
\[
M_{t+1}=M_t\cup\{\pi_t\},
\]
where \(\pi_t\) records a certificate produced at time \(t\).

In regulated finance, however, persistent proof memory should record more than
the certificate alone.  It should record both the lineage of the trace and the
lineage of the policy system under which the trace was certified.  A financial
trace is not certified against a timeless policy; it is certified against a
particular version of the policy system in force at certification time.

Thus a persistent-memory certifier may store records of the form
\[
(\tau_t,\pi_t,\Pi_t,\mathrm{source}_t,\mathrm{version}_t,\mathrm{effective\ date}_t),
\]
where \(\tau_t\) is the certified trace, \(\pi_t\) is the certificate,
\(\Pi_t\) is the policy system used for certification, and the remaining fields
record the legal, regulatory, mandate, or internal-policy lineage of \(\Pi_t\).

This distinction matters because financial requirements evolve.  Laws,
regulations, supervisory guidance, client mandates, fund prospectuses, and
internal risk rules may change over time.  A later auditor may ask two different
questions:
\[
\text{Was } \tau_t \text{ permissible under } \Pi_t?
\]
and
\[
\text{Would } \tau_t \text{ remain permissible under a later policy system }
\Pi_{t'}?
\]
These are distinct questions.  The first is historical verification; the second
is re-certification under policy change.

Thus persistent proof memory may be decomposed as
\[
\text{Persistent Proof Memory}
=
\text{Trace Lineage}
+
\text{Policy Lineage}.
\]
This decomposition becomes especially important when \(\Pi\) has internal
structure: a later audit may need to know not only which trace was certified,
but also which policy layers, policy versions, access permissions, and evidence
obligations were used to certify it.
Trace lineage records how the action, report, recommendation, or trade was
produced: retrieval, tool use, evidence, computation, claim generation, and
execution conditions.  Policy lineage records which policy system was applied:
legal requirements, regulatory guidance, fund mandate, client policy, internal
risk rule, execution-control rule, version, source, and effective date.

\paragraph{Layered policy systems and proof obligations.}
The policy system \(\Pi\) should not be viewed as a flat list of rules. In real
financial and access-control environments, policy is often layered, partially
ordered, and versioned. Some components are hard constraints; some are inherited
permissions; some are temporal; some require escalation; and some depend on
evidence, prior trace history, or policy version.

Schematic examples include
\[
\Pi
=
\Pi_{\mathrm{law}}
\wedge
\Pi_{\mathrm{mandate}}
\wedge
\Pi_{\mathrm{client}}
\wedge
\Pi_{\mathrm{risk}}
\wedge
\Pi_{\mathrm{access}}
\wedge
\Pi_{\mathrm{privacy}}
\wedge
\Pi_{\mathrm{execution}}.
\]
This expression is schematic rather than literal: the layers need not form a
simple linear order. For example, legal or regulatory constraints may dominate
client-specific preferences; fund mandates may specialize institutional rules;
access-control policies may be inherited through roles or groups; and execution
rules may depend on the current market state, account state, or unresolved
certification obligations.

Access control is a particularly clear example. A financial agent may need to
certify that a user is authorized to access a source, that the agent is permitted
to call a tool, that the retrieved data may be used for the stated purpose, that
sensitive information may not leave a protected environment, that the final
output may contain the derived information, and that the trace may be stored in
logs. Some of these obligations are local, such as whether a user has permission
to read a document. Others are trace-dependent, such as whether a source was
retrieved for one purpose and later used for another. Still others are
policy-lineage dependent, because an access may have been allowed under policy
version \(\Pi_t\) but may require re-certification under a later policy system
\(\Pi_{t'}\).

Under a layered policy system, certification naturally decomposes into layered
proof obligations. If
\[
\Pi=\bigwedge_{\ell=0}^{d}\Pi_\ell,
\]
then a certificate may be represented schematically as
\(
\pi=(\pi_0,\pi_1,\ldots,\pi_d),
\)
where \(\pi_\ell\) certifies the obligations induced by layer \(\Pi_\ell\). For
example, an access-control component may certify authorized retrieval, an
evidence component may certify source binding, an accounting component may
certify period and unit consistency, a risk component may certify mandate or
exposure constraints, and a privacy component may certify non-disclosure or
redaction obligations.

Thus policy layers and proof-memory classes play different roles. The policy system determines
what must be certified; proof memory determines what can be remembered, replayed, and certified.
Proof memory is the bridge between them: it records which policy-layer obligations have already
been satisfied, under which policy version, and with which supporting evidence.

This motivates persistent-memory certifiers of class
\[
M_\Pi^{(3)}.
\]
Examples include immutable audit trails, blockchain commitments, zero-knowledge
proof certificates, and versioned policy ledgers.

This suggests a hierarchy of proof memory classes:
\[
\text{Finite}
<
\text{Structured}
<
\text{Persistent}.
\]
Richer proof memory may permit certification of richer permissible languages.

\subsection{Self-Approval Is Not Certification}
\label{app:example-self-approval}

A final example clarifies why the proposal and certification roles should not be collapsed.
Suppose an autonomous trading agent proposes a trace $\tau$ and also outputs the statement

\[
\pi_{\mathrm{self}}=\text{``this trace is permissible because I certify it.''}
\]

This statement should not by itself count as a valid certificate. It is only a self-assertion by the
same system whose behavior is being judged. If such self-approval were sufficient, then the
distinction between generation and certification would collapse: a generator could make a trace
executable merely by claiming that it is executable.

A simple governance rule is therefore:

\[
R_{\mathrm{ind}}:\quad
\text{A trace may execute only if its certificate is verified by a certifying role
distinct from the proposal itself.}
\]

Here ``distinct'' need not mean a different physical machine. It means a logically separate
certification role: for example, an external policy engine, a proof checker, a compliance contract,
a verifier with access to policy $\Pi$, or a human escalation process. The generator may provide
evidence or a proposed proof object, but the certificate must be checkable by the certifier.

Under this rule, the trace
\[
(\tau,\pi_{\mathrm{self}})
\]
does not execute merely because the agent approves itself. Instead, the certifier must verify that
there exists a valid certificate
\[
\pi \quad \text{such that} \quad \Pi \vdash \tau:\mathrm{Permitted}.
\]
If no such certificate can be verified, the outcome is rejection or escalation:
\[
M_\Pi(\tau)\in \{\bot,\mathrm{Escalate}\}.
\]

This example illustrates the principle
\[
\text{Self-Approval} \not\Rightarrow \text{Permissibility}.
\]

The example is not meant to prove a formal paradox. Its purpose is to show why trustworthy
financial agents should separate the authority to propose a trace from the authority to certify
that the trace may execute. More generally, policies that refer to certification, delegation, or
authority may create self-referential obligations. Such obligations motivate the incompleteness
questions discussed later in Section~7.

\begin{figure}[!t]
    \centering
    \includegraphics[width=0.9\textwidth]{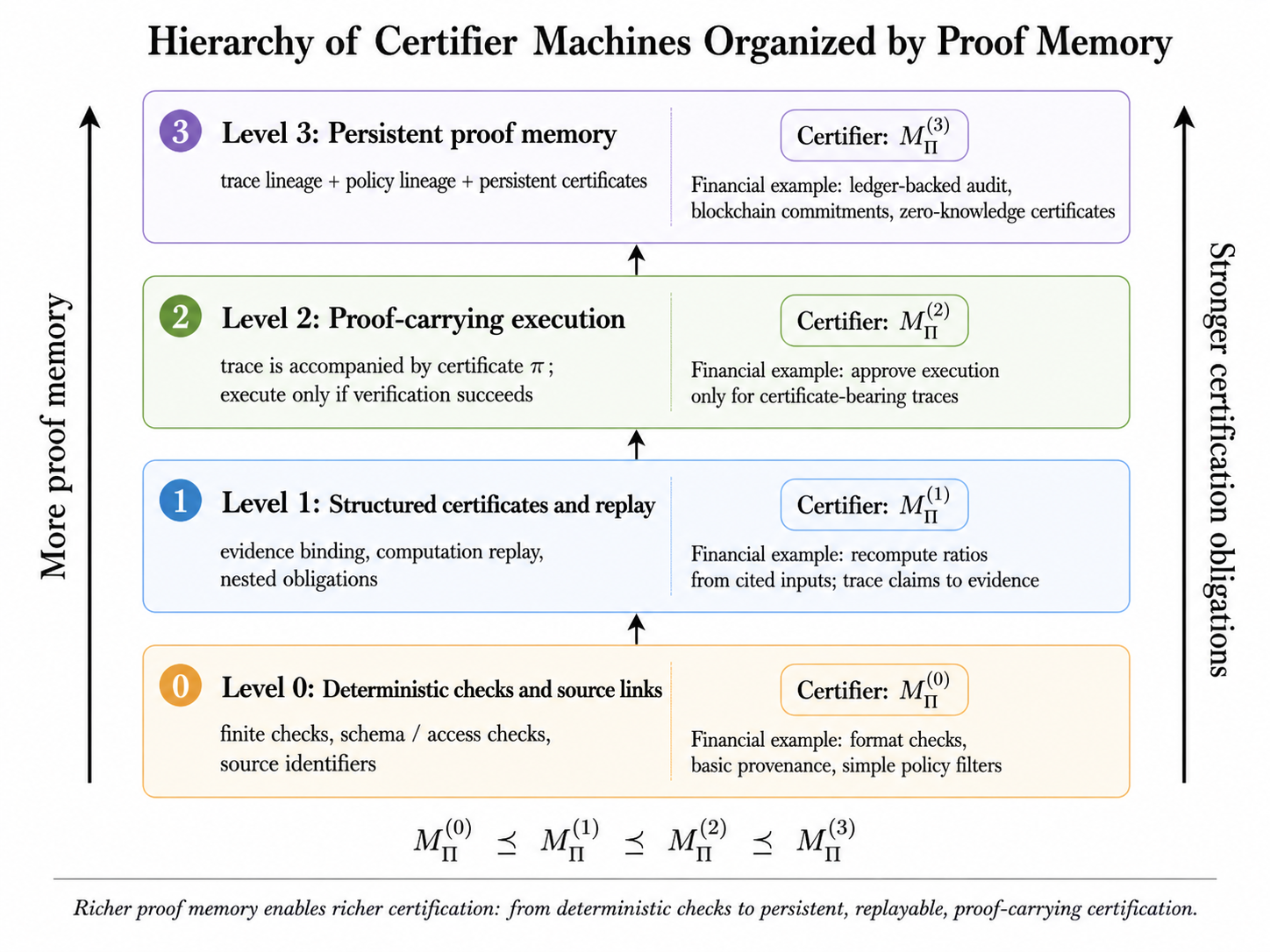} 
     \caption{
    \textbf{Proof-memory classes of certifiers.}
    The examples in Appendix~\ref{app:illustrative-examples}
motivate a progression of certifier classes ordered by proof-memory support, from local checks to
persistent proof memory:
$M^{(0)}_\Pi \preceq M^{(1)}_\Pi \preceq M^{(2)}_\Pi \preceq M^{(3)}_\Pi$.
Level 0 performs deterministic checks and attaches source links; Level 1 constructs structured
certificates and supports replay of evidence and computations; Level 2 supports proof-carrying
execution, where a trace executes only when accompanied by a valid certificate; and Level 3 maintains
persistent proof memory, including trace lineage, policy lineage, and durable certification records.
The hierarchy is not a claim that all implementations fall into exactly four classes; it illustrates how
richer proof memory supports richer trace-level certification.}
    \label{fig:proof-memory-hierarchy}
\end{figure}

\subsection{Examples as Separation and Proof-Memory Witnesses}
\label{app:examples-as-witnesses}

The examples in Appendix~\ref{app:illustrative-examples} have a cumulative structure.
Appendices~\ref{app:example-impermissible-strategy}--\ref{app:example-optimality}
witness the action-level permissibility gap. Appendix~\ref{app:example-noncompositionality}
shows why locally permissible actions may compose into an impermissible trace.
Appendix~\ref{app:example-correct-answer-impermissible-trace} motivates evidence binding,
claim support, and computation replay. Appendix~\ref{app:example-escalation} shows why
escalation may itself be the permissible continuation under uncertainty.
Appendices~\ref{app:example-proof-carrying-execution}--\ref{app:example-persistent-proof-memory}
motivate proof-carrying execution, policy lineage, and persistent proof memory.
Section~\ref{sec:privacy-zk} develops the privacy and zero-knowledge trace-certification setting.

Thus the examples are not isolated anecdotes; they are witnesses for increasingly rich certification requirements. They point to the same paradigm:
\[
\text{language models propose; Permissibility Machines certify.}
\]
Certification mediates execution.

These examples witness recurring separation phenomena:
\[
\text{Profitability}
\not\Rightarrow
\text{Permissibility},
\]
\[
\text{Optimality}
\not\Rightarrow
\text{Permissibility},
\]
\[
\text{Permissible Actions} \not\Rightarrow 
\text{Permissible Trace},
\]
\[
\text{Correct Answer}
\not\Rightarrow
\text{Permissible Trace},
\]
\[
\text{High Confidence}
\not\Rightarrow
\text{Permission},
\]
\[
\text{Self-Approval}
\not\Rightarrow
\text{Permissibility}.
\]
They also motivate the certifier hierarchy
\[
M_\Pi^{(0)}
\preceq
M_\Pi^{(1)}
\preceq
M_\Pi^{(2)}
\preceq
M_\Pi^{(3)}.
\]

Fig.~\ref{fig:proof-memory-hierarchy} summarizes this progression. The figure
should be read as an organizational hierarchy rather than as a claim that all
certifier implementations must fall into exactly four rigid classes. Its purpose
is to make explicit the role of proof memory as the organizing axis of
certification. At the lowest level, a certifier performs deterministic checks
and attaches source links. At richer levels, the certifier constructs structured
certificates, supports replay of evidence and computations, conditions execution
on proof objects, and eventually maintains persistent records for audit and
re-certification.

The central point is that proof memory is not merely storage. It records the
proof-relevant history needed to certify a trace: prior actions, current state,
retrieved evidence, computations, approvals, policy versions, and unresolved
obligations. As the policy system becomes more layered and trace-dependent, the
certifier must preserve more of this history in order to determine whether a
generated trace may execute.

Richer proof memory may permit certification of richer permissible languages. This motivates the heuristic principle:

\begin{heprinciple}[Memory--Expressiveness Tradeoff]
\[ \text{If~~~}
M_\Pi^{(i)}
\preceq
M_\Pi^{(j)},
\text{~~~then~~~}
L_{\mathrm{cert}}(M_\Pi^{(i)})
\subseteq
L_{\mathrm{cert}}(M_\Pi^{(j)}).
\]
\end{heprinciple}

The examples 
suggest both
\(
C\setminus P\neq \varnothing
\)
and the need for increasingly expressive certifiers.
They all point to the same paradigm:
\[
\boxed{
\text{Language Models propose;}
\quad
\text{Permissibility Machines certify}
}.
\]
Certification mediates execution.

These examples also suggest a broader hypothesis:
generation constraints may reduce violations in
\(
C\setminus P,
\)
but may not eliminate the need for certification. This motivates asking whether constrained generation can subsume certification at all.

\section{Related Work and Positioning}
\label{app:related-work}

This appendix positions the proposed framework relative to nearby areas. The contribution is not
that certification, monitoring, or proof checking are new. The contribution is to identify certified
execution traces as an organizing object for trustworthy AI agents.

\paragraph{Positioning.}
This paper sits at the intersection of chain-of-thought monitorability, LLM-agent safety, runtime
enforcement, proof-carrying code, and information-flow security. CoT monitorability studies whether
reasoning trajectories can be inspected for signs of misbehavior~\citep{korbak2025chain,
openai2025cotmonitorability}. Agent benchmarks and attacks expose failures arising from tool use,
untrusted data, environment interaction, and multi-step behavior~\citep{ruan2024toolemu,
debenedetti2024agentdojo,yang2024watchagents,zhang2024agentsafetybench}. Runtime enforcement
intervenes during execution~\citep{havelund2004jpaX,xiang2024guardagent}; proof-carrying code
checks safety certificates for code~\citep{necula1997proof}; and information-flow security studies
end-to-end confidentiality and data-flow constraints~\citep{sabelfeld2003informationflow}. Certified
traces combine these ideas around a policy-governed execution rule: no certificate, no execution.

\begin{table}[h]
\centering
\small
\begin{tabular}{p{0.23\linewidth}p{0.30\linewidth}p{0.38\linewidth}}
\toprule
Area & Main object & Relation to certified traces \\
\midrule
Output guardrails & Final answer or response & Useful but misses trace-level violations. \\
CoT monitorability & Reasoning trajectory & One observation channel inside a broader execution trace. \\
Agent-safety benchmarks & Tool-use failures and adversarial tasks & Expose failures that certified traces aim to govern. \\
Runtime monitoring & Execution-time rule violations & Detects or intervenes; certification authorizes before execution. \\
Proof-carrying code & Code plus safety proof & Formal ancestor of certificate-bearing execution. \\
Information-flow security & Data-flow constraints & Provides policy models for privacy and access-control traces. \\
\bottomrule
\end{tabular}
\caption{Positioning of certified traces relative to nearby research areas.}
\label{tab:positioning}
\end{table}


\paragraph{Proof-carrying code.}
Proof-carrying code requires executable code to carry a checkable safety proof before execution
~\citep{necula1997proof}. Our framework adapts this idea from code artifacts to AI-agent traces.
The certificate is attached not merely to a program, but to a generated sequence of retrievals, tool
calls, computations, memory updates, evidence bindings, and execution conditions.


\paragraph{Runtime monitoring and formal verification.}
Runtime monitors inspect system behavior against specifications, and formal verification studies
whether systems satisfy formal properties~\citep{havelund2004jpaX,leucker2009runtime}.
Permissibility Machines inherit this verification spirit, but focus on the agent-specific boundary
between generated traces and execution-authorized traces under policy systems.

\paragraph{AI agents and tool-use safety.}
Recent AI-agent systems can call tools, retrieve information, modify external state, and coordinate
multi-step workflows. This expands the safety object from final output to execution trace. Certified
traces provide a way to ask whether the full agent workflow, not merely the final answer, satisfies
the policy system.

\paragraph{Tool-using and scientific agents.}
Toolformer and ReAct show how language models can integrate tool calls and interleave reasoning
with acting~\citep{schick2023toolformer,yao2023react}. Scientific-agent systems such as ChemCrow
demonstrate the same shift in high-stakes domains, where agents use expert tools for chemistry,
drug discovery, and materials design~\citep{bran2024chemcrow}. Recent discussions of AI scientists
emphasize that autonomous scientific agents require stronger safeguarding, especially around tool
use, action spaces, and downstream impact~\citep{tang2025aiscientists}. These works motivate our
claim that the safety object is not merely the final answer, but the execution trace.

\paragraph{Chain-of-thought monitorability and reasoning traces.}
Recent work frames chain-of-thought monitorability as a new but fragile opportunity for AI safety:
visible reasoning may allow monitors to detect intent to misbehave, but the opportunity may change
as models, training methods, and reasoning formats evolve~\citep{korbak2025chain}. OpenAI's
monitorability evaluations define monitorability as the ability of a monitor to predict properties of
an agent's behavior from observations such as actions, outputs, activations, or chain-of-thought
signals~\citep{openai2025cotmonitorability}. Our framework is complementary. We treat
chain-of-thought as one possible component of a broader execution trace, together with retrievals,
tool calls, data flows, computations, memory updates, logs, and execution conditions. Visibility
alone is not certification: a trace may be visible but unsupported, unauthorized, private, or
unreplayable. Certified traces ask for checkable witnesses of permissibility, not merely inspectable
reasoning.

\paragraph{Agent safety and tool use.}
Agent-safety benchmarks increasingly study agents that combine reasoning with external tool calls.
For example, AgentDojo evaluates prompt-injection attacks and defenses for agents executing tools
over untrusted data~\citep{debenedetti2024agentdojo}. Such benchmarks expose precisely the kind
of trace-level failures that output checking misses: unsafe retrieval, malicious tool outputs, and
policy-violating actions. Our paper is complementary: rather than proposing another benchmark, it
asks what formal object should govern whether such traces may execute.

\paragraph{Layered attack surfaces for agentic systems.}
Recent survey work proposes a Layered Attack Surface Model for LLM-based agents, separating
threats across foundation, cognitive, memory, tool-execution, multi-agent, ecosystem, and
governance layers, and adding a temporal axis for attacks that persist across sessions or manifest
slowly~\citep{chu2026lasm}. This taxonomy is complementary to our framework. LASM asks where
agentic threats arise; certified traces ask when a generated trace may be authorized for execution
under a policy system.

\paragraph{Runtime guardrails and enforcement for agents.}
Recent work also studies guardrails and runtime enforcement for LLM agents. GuardAgent, for
example, oversees a target LLM agent by checking whether its behavior satisfies user-defined guard
requests and by translating those requests into executable checking code~\citep{xiang2024guardagent}.
These systems are closely related to our motivation. The distinction is that our paper makes
certification a precondition for execution and studies the induced certification boundary, rather than
treating guardrails only as runtime checks or post-hoc filters.

\paragraph{AI governance and auditability.}
Governance frameworks emphasize accountability, transparency, risk management, and auditability.
Certified traces translate these goals into technical objects: trace languages, certificates, proof
memory, policy lineage, and certification boundaries.

\paragraph{From model materials to trace materials.}
Recent licensing work for machine-learning models introduces the notion of \emph{Model Materials}:
the model architecture, parameters, data, documentation, software, evaluation assets, and related
artifacts distributed with a model. This is a useful analogy for certified execution. Model-release
governance asks what materials are needed to understand, reuse, or audit a model distribution.
Agent-execution governance asks what trace materials are needed to authorize an action: evidence,
tool calls, source access, approvals, computations, credentials, policy versions, and execution
conditions. Open model materials clarify what is being released; certified trace materials clarify
what is being executed.

\paragraph{Privacy and zero-knowledge proofs.}
Privacy-preserving verification and zero-knowledge proofs make it possible to prove selected
properties without revealing sensitive data. In this paper, they address the privacy--observability
tension: trace certification needs evidence, but raw traces may contain private information.

\paragraph{Financial AI systems and benchmarks.}
Financial AI systems and benchmarks motivate the framework because finance combines structured
data, regulations, audit requirements, privacy constraints, and execution risks. The proposed theory
is broader than finance, but finance provides concrete examples where trace-level certification is especially visible. Financial-domain agent systems provide concrete testbeds for trace certification. 

\paragraph{Industry movement toward financial agents.}
Recent industry announcements suggest that financial AI agents are moving from pilots toward
production workflows. OpenAI and PwC describe agents for CFO workflows such as procurement,
payments, treasury, tax, and close, with an emphasis on controls and human oversight
~\citep{openai2026pwcfinance}. American Express's ACE initiative points toward agentic commerce
and payment credentials for verified agents~\citep{amex2026ace}. Anthropic's financial-services
agents connect to market data, research platforms, and internal systems under governed access
controls~\citep{anthropic2026financeagents}. ICBC reports AI agent systems and large-model
infrastructure for settlement finance, credit risk, and risk-control workflows~\citep{icbc2025interimqa,
icbc2026annualresults}. These developments motivate our central question: when an agent enters a
real financial workflow, what certificate authorizes execution?

XBRL-Agent
uses LLM agents with retrievers and calculators for XBRL report analysis, exposing both the
promise and the reliability limits of tool-augmented financial analysis~\citep{han2024xbrlagent}.
Our framework abstracts these concerns into certificate components for source access, tag
selection, period alignment, unit normalization, computation replay, and claim support.

\section{Guiding Principles and Remaining Open Problems}
\label{app:open-problems}

This appendix collects guiding principles that follow from the proposal--certification--execution
framework and identifies remaining open problems.

\subsection{Guiding Principles}

We summarize the guiding principles as follows.

\paragraph{Principle 1: Computability is not permissibility.}
A system may be able to represent, generate, or submit an action without that
action being permitted:
\(
P\subsetneq C.
\)

\paragraph{Principle 2: Permissibility is trace-dependent.}
Permissibility is not only an action-level property. It may depend on the trace
in which actions occur: state, timing, evidence, computation, source access, and
policy history.

\paragraph{Principle 3: Generation and certification are distinct roles.}
A generating machine proposes candidate traces, while a Permissibility Machine
certifies whether a generated trace may execute under a policy system \(\Pi\).

\paragraph{Principle 4: No certificate, no execution.}
Execution is not a direct consequence of generation. A trace becomes executable
only when accompanied by a certificate witnessing permissibility under \(\Pi\).

\paragraph{Principle 5: Proof memory supports certification.}
Persistent proof memory records trace lineage and policy lineage, allowing
certifiers to replay, audit, and re-certify trace-level obligations.

\paragraph{Principle 6: No self-certification.}
A high-stakes agent may propose evidence for its trace, but it should not be the sole authority
certifying that its own proposed trace is permissible.

\subsection{Open Problems}

\paragraph{Open Problem 1: Certifiable trace languages.}
Which policy-induced trace languages $L_\Pi \subseteq \Sigma^*$ admit efficient sound
certification? The answer should depend on the policy model: finite-state workflow rules
may admit lightweight certificates, counter and temporal policies may require proof memory,
and information-flow or proof-calculus policies may require richer evidence, provenance,
or cryptographic witnesses. A theory of certifiable trace languages should characterize
the tradeoff among policy expressiveness, proof-memory requirements, verification time,
and certificate size.

\paragraph{Open Problem 2: Proof-memory lower bounds.}
What proof-memory capacity is necessary to certify a given class of trace-dependent policies? Can
one prove Myhill--Nerode-style lower bounds showing that certain policies require distinguishing
a minimum number of proof-relevant history classes?

\paragraph{Open Problem 3: Certification rate-distortion.}
What is the best achievable certification boundary under constraints on memory, latency, privacy,
or verification cost? Can one derive rate--distortion laws relating proof-memory capacity to
certified-but-impermissible risk and missed-certification risk?

\paragraph{Open Problem 4: Policy drift and re-certification.}
How should certificates be maintained, invalidated, or updated when policy systems evolve? What
is the right measure of distance between policy versions, and how does this distance determine
re-certification burden?

\paragraph{Open Problem 5: Privacy-preserving proof memory.}
How can proof memory support auditability without exposing sensitive trace contents? What trace
properties can be certified using commitments, selective disclosure, or zero-knowledge proofs?

\paragraph{Open Problem 6: Compositional certification for multi-agent systems.}
When do certificates for individual agents, tools, or subtasks compose into certificates for
multi-agent traces? What additional obligations arise from delegation, communication, and shared
memory?

\paragraph{Open Problem 7: Benchmarks for certified agents.}
How should the community evaluate trace-level certification rather than output accuracy alone?
What datasets, tasks, metrics, and failure cases are needed to measure generated permissibility risk,
certified-but-impermissible risk, missed certification, executable yield, and certification cost?



\end{document}